\documentclass[showpacs,pra,twocolumn,superscriptaddress,aps]{revtex4-1}

\usepackage{exscale}
\usepackage{bbm}
\usepackage{graphicx}
\usepackage{latexsym}
\usepackage{times}
\usepackage[T1]{fontenc}
\usepackage{enumerate}
\usepackage{bbold}
\usepackage{color}
\usepackage{mathtools}
\usepackage{changes}
\usepackage[normalem]{ulem}
\usepackage{amsfonts,amsmath,amssymb,amsthm,dsfont}
\usepackage[colorlinks=true,citecolor=blue,urlcolor=blue]{hyperref}

\usepackage{physics}

\theoremstyle{plain}
\newtheorem{thm}{Theorem}
\newtheorem{lem}{Lemma}

\newtheorem{cor}{Corollary}

\newcommand{\de}[1]{\left( #1 \right)}

\newcommand{\DE}[1]{\left\{#1\right\}}
\newcommand{\id}{\mathbbm{1}} 

\begin{abstract}
Experimental detection of multipartite Bell nonlocality faces many challenges. The increase in the size of the system demands a higher fidelity of state preparation, a higher detection efficiency, and a larger number of correlators that need to be estimated. These issues are even more pronounced if one wishes to detect a more restrictive variant of Bell nonlocality called genuine multipartite nonlocality (GMNL). In this work, we introduce two classes of multipartite Bell inequalities that are capable of detecting GMNL correlations and address different experimental hurdles in GMNL detection. The first class consists of the number of correlators that scales linearly with the number of parties. The second class increases the number of correlators and, in return, provides higher noise robustness of GMNL detection. Both classes of inequalities are constructed to detect GMNL in graph states, and in fact both can be used for self-testing of graph states and the performed measurements.

\end{abstract}

\begin{document}
\title{Genuine multipartite nonlocality detection in graph states}
\author{Owidiusz Makuta}
\affiliation{Instituut-Lorentz, Universiteit Leiden, P.O. Box 9506, 2300 RA Leiden, The Netherlands}
\affiliation{$\langle \text{aQa}^\text{L} \rangle$ Applied Quantum Algorithms Leiden, The Netherlands}

\author{Flavio Baccari}
\affiliation{Dipartimento di Fisica e Astronomia \say{G. Galilei} \& Padua Quantum Technologies Research Center, Università degli Studi di Padova, I-35131 Padova, Italy}

\affiliation{INFN, Sezione di Padova, via Marzolo 8, I-35131, Padova, Italy}

\author{Remigiusz Augusiak}
\affiliation{Center for Quantum-Enabled Computing, Center for Theoretical Physics, Polish Academy of Sciences, Aleja Lotnik\'{o}w 32/46, 02-668 Warsaw, Poland}

\maketitle
\section{Introduction}

Bell nonlocality is one of the most fundamental features of quantum mechanics. Its detection remains the strongest test of the quantumness of an experimental setup, as it relies solely on general principles governing probabilities and requires no assumptions about the internal functioning of the measurement devices. As a result, Bell nonlocality is a resource for several applications within the device-independent framework such as quantum cryptography \cite{Zapatero_2023}, self-testing \cite{Supic2020selftestingof}, distributed computing \cite{Buhrman_2010}, secure delegated computing \cite{Reichardt2013-tk} and randomness amplification \cite{Pironio_2010}.

However, the extraordinary strength of Bell nonlocality as a certificate of quantum behavior demands an equally extraordinary experimental setup for its observation. The lack of assumptions about the measurement devices naturally makes their detection more susceptible to noise. Then, its proper implementation requires spatial separation between measurements, which requires either a large distance between the measurement devices or ultra-fast measurements. On top of that, a high detection rate is also a strict requirement, as any missed detections must be counted towards the total tally. Due to these challenges, over 50 years passed between the theoretical discovery of Bell nonlocality \cite{PhysicsPhysiqueFizika.1.195} and its complete experimental verification \cite{Hensen2015-en, PhysRevLett.115.250401, PhysRevLett.115.250402, PhysRevLett.119.010402, Storz2023-tj}.

While the most basic Bell scenario requires only two measurement devices, many applications listed above either benefit from or require a higher number of them. This led researchers to focus on multipartite Bell scenarios, in which, depending on the motivation, one can identify many possible variants of Bell nonlocality \cite{ PhysRevLett.103.090503, PhysRevLett.127.200401,PhysRevD.35.3066, PhysRevLett.109.070401,PhysRevA.88.014102}.  One such variant, called \textit{genuine multipartite nonlocality} (GMNL) \cite{PhysRevD.35.3066, PhysRevLett.109.070401,PhysRevA.88.014102}
describes nonlocal correlations that span the full system size.
Due to this, GMNL serves as a much stronger certificate, as it requires global correlations in the measurement data, making any applications of it less susceptible to attacks by malicious agents.

Although many potential advantages of Bell nonlocality lie in the realm of multipartite scenarios, with few exceptions \cite{PhysRevLett.129.150402, Hu_2025}, most of the research up to this point has focused on Bell nonlocality detection in the simplest scenario with two measurement devices. This is because in the multipartite setting every challenge in nonlocality detection becomes more severe. For example, the increased number of measurements increases the likelihood of any one of them producing a faulty result; hence, high-fidelity experimental equipment is necessary for a multipartite Bell test. Moreover, to witness multipartite variants like GMNL, the number of correlators that need to be estimated grows with the system size, which further amplifies the effects of noise on the data.

However, from a theoretical perspective, we also face an increased difficulty when multiple measurement devices are considered. This is due to the fact that the probability space of all possible measurement statistics grows exponentially with the number of measurement devices, making their full characterisation virtually impossible. Consequently, a possible way out is to design Bell inequalities which 
are tailored to the specific quantum system under study. Several such constructions have been provided recently, such as those for symmetric states \cite{TURA2015370, Curchod_2019}, graph states \cite{PhysRevLett.124.020402, PhysRevLett.95.120405, article, zhao2026certifyingoptimaldeviceindependentquantum}, or stabilizer codes \cite{PhysRevLett.125.260507, PhysRevResearch.4.043215, vallée2025tailoringbellinequalitiesqudit, Bulchandani_2023}. 

In particular, Ref. \cite{PhysRevLett.124.020402} introduced a class of Bell inequalities tailored to $N$-qubit graph states, the so-called Baccari-Augusiak-Šupić-Tura-Acín (BASTA) inequality.
Notably, testing nonlocality with these inequalities requires determining a number of expected values that scales linearly with the number of qubits, making them good candidates for experimental implementation \cite{Wu2021robust,Yang_2022}. What is more, since they enable self-testing, they are well-suited for real-life applications of nonlocality. However, the BASTA inequalities, while being 'experimentally friendly', turn out to be incapable of detecting GMNL. On the other hand, the known way of detecting GMNL in graph states \cite{Almeida_2010} has the opposite problem: it requires a number of measurement settings that scales linearly with the number of parties \cite{Makuta_2025}, which makes the number of correlators that need to be estimated even higher. 

The main aim of this work is to adapt BASTA Bell inequalities to make them capable of detecting genuine nonlocality, while retaining the advantages such as low correlation number and the self-testing property.
The paper is structured as follows: Sec. \ref{sec:pre} introduces graph states, Bell nonlocality and the BASTA inequality. Sec. \ref{sec:res} covers the results of this work: a construction of an inequality detecting GMNL based on the BASTA inequality, its robust self-testing, and a second construction of a GMNL-detecting inequality based on the Svetlichny inequality. We conclude the manuscript with Sec. \ref{sec:con} where we summarise the main conclusions of this work and list remaining open problems.

\section{Preliminaries}\label{sec:pre}
In this section, we introduce definitions and concepts required for the derivation and understanding of our results: graph states and Bell nonlocality.

\subsection{Graph states}
To introduce graph states, we first need to define a couple of notions from graph theory. A graph is a tuple $G=(V,E)$ of the set of vertices $V=\{1,\dots,n\}=[n]$ and the set of edges $E$ connecting them. Next, the \textit{neighbourhood} $N_{i}$ of the vertex $i$ is the set of vertices connected to $i$ by an edge
\begin{equation}
N_{i} \coloneq \{j\in V\;| \;(i,j)\in E\}.
\end{equation}
Notably, in this paper, we assume that the considered graph $G$ is connected, i.e., we can get from any $i\in V$ to any $j\in V$ by traversing the edges of $G$.

With each graph $G$ satisfying the above constraints, one can associate the corresponding \textit{graph state} $\ket{G}$. Let us consider a Hilbert space $\mathcal{H}=(\mathbb{C}^{2})^{\otimes n}$; we associate a subspace $\mathcal{H}_{i}=\mathbb{C}^{2}$ (called qubit) with a vertex $i\in V$. Then, a graph state $\ket{G}$ of the graph $G$ is defined as a unique state from $\mathcal{H}=(\mathbb{C}^{2})^{\otimes n}$ that satisfies
\begin{equation}
 g_{i} \ket{G} = \ket{G},\quad \forall_{i\in[n]},
\end{equation}
where $g_{i}$, called the stabilizing operators of $\ket{G}$, are defined as
\begin{equation}\label{eq: generators}
g_{i} \coloneq X_{i} \prod_{j\in N_{i}}Z_{j},
\end{equation}
and $X_{i},Z_{i}$ are operators acting on $\mathcal{H}_{i}$ with Pauli matrices $X$ and $Z$, respectively, and with the identity $\mathbb{1}$ on every other subspace $\mathcal{H}_{j}$, $j\neq i$. Notice that the stabilizing operators $g_{i}$ encode the information about the graph $G$, as they act with $X$ on the vertex $i$ and with $Z$ on any vertex in its neighbourhood.

Importantly, some graph states $\ket{G},\ket{G'}$ are equivalent to each other up to local basis change, even though they originate from different graphs. In fact, if we can transform a graph $G$ into $G'$ by a series of graph transformations called \textit{local complementation}, it was shown that the resulting graph states $\ket{G},\ket{G'}$ are equivalent up to local unitaries \cite{PhysRevA.69.022316}. The local complementation is an operation acting on vertex $k$ and for each pair $i,j\in N_{k}$ it adds an edge $(i,j)$ if $(i,j)\notin E$, and it removes an edge $(i,j)$ if $(i,j)\in E$.

\subsection{Bell nonlocality}

Let us consider an $n$-partite scenario, where each party $i$ has access to measurement device that, after choosing a measurement setting $x_{i}$, produce an outcome $a_{i}$. Moreover, let us consider a division of the set $[n]$ into two disjoint sets $Q$ and $\overline{Q}$ called a bipartition $Q|\overline{Q}$. For notational simplicity, let $\mathbf{a}_{Q}=\{a_{i}\}_{i\in Q}$ and $\mathbf{x}_{Q}=\{x_{i}\}_{i\in Q}$ represent the outcomes and settings for parties in the set $Q$. Accordingly, $\mathbf{a}$ and $\mathbf{x}$ denote the outcomes and settings for the full set of parties $[n]$.

A \textit{behaviour} $\vec{p}=\{p(\mathbf{a}|\mathbf{x})\}_{\mathbf{a},\mathbf{x}}$ describing this scenario is said to be local if it admits the following decomposition
\begin{equation}\label{eq:local}
p(\mathbf{a}|\mathbf{x})=\sum_{\lambda}q_{\lambda}\prod_{i\in [n]} p(a_{i}|x_{i},\lambda),
\end{equation}
where $q_{\lambda}$ is a probability distribution and $\lambda$ is called a local hidden variable. Then, any behaviour that is not local is called Bell nonlocal (or simply nonlocal).

While this definition extends the traditional notion of Bell nonlocality to the multipartite regime, the latter exhibits a significantly richer structure, offering distinct forms of nonlocality relevant to specific research objectives. Notably, if we wish to capture the notion of "global" nonlocality among all $n$ parties, then we should consider the notion of \textit{genuine multipartite nonlocality} (GMNL) \cite{PhysRevD.35.3066, PhysRevLett.109.070401,PhysRevA.88.014102}. 

To this end, let us first consider $\vec{p}$ that are local across a bipartition $Q|\overline{Q}$
\begin{equation}\label{eq:local_bip}
p(\mathbf{a}|\mathbf{x})=\sum_{\lambda}q_{\lambda}p(\mathbf{a}_{Q}|\mathbf{x}_{Q},\lambda)p(\mathbf{a}_{\overline{Q}}|\mathbf{x}_{\overline{Q}},\lambda),
\end{equation}
where we demand that the marginal probability distributions $p(\mathbf{a}_{Q}|\mathbf{x}_{Q},\lambda)$ are non-signalling \cite{PhysRevLett.109.070401,PhysRevA.88.014102}. Any behaviour that is not local across $Q|\overline{Q}$ is called Bell nonlocal (or simply nonlocal) across $Q|\overline{Q}$. Clearly, a behaviour that is nonlocal across $Q|\overline{Q}$ is also nonlocal in the sense of Eq. \eqref{eq:local}.

Finally, a behaviour $\vec{p}$ is said to be \textit{bi-local} if it admits
\begin{equation}\label{eq:bilocal}
p(\mathbf{a}|\mathbf{x})= \sum_{\substack{Q\subset[n]\\ |Q|\leqslant \lfloor n/2\rfloor}}r_{Q}\sum_{\lambda}q_{\lambda,Q}p(\mathbf{a}_{Q}|\mathbf{x}_{Q},\lambda)p(\mathbf{a}_{\overline{Q}}|\mathbf{x}_{\overline{Q}},\lambda),
\end{equation}
where $r_{Q}$ and $q_{\lambda,Q}$ are probability distributions, $\lfloor\cdot \rfloor$ is the floor function, and we again demand that the marginal probability distributions $p(\mathbf{a}_{Q}|\mathbf{x}_{Q},\lambda)$ are non-signalling. Then, any non-bi-local behaviour, that is, one that does not admit a representation in terms of bi-local models \eqref{eq:bilocal} is termed GMNL. This notion has a very clear interpretation: it excludes behaviours that are local across some bipartition $Q|\overline{Q}$, but also probabilistic mixtures of such behaviours. In this sense, GMNL characterises behaviour that is truly $n$-partite. 

\subsection{Bell inequalities for graph states}

The most common way of detecting nonlocality is through the violation of a Bell inequality. The latter takes the form of a linear expression in $\vec{p}$, also called a \textit{Bell functional}, which is bounded from above by its maximal value over local models. More specifically, a generic form of a Bell inequality reads  
\begin{equation}
    I:=\vec{\alpha}\cdot\vec{p}\leq \beta_L,
\end{equation}
where $\vec{\alpha}$ is a given real multidimensional vector specifying the Bell inequality, $(\cdot)$ is 
the standard scalar product and $\beta_L=\max_{\mathrm{local}}I$ is referred to as the classical bound and is the maximal value of the Bell functional $I$ over the local models \eqref{eq:local}.

The simplest example of a Bell inequality is the Clauser–Horne–Shimony–Holt (CHSH) inequality \cite{PhysRevLett.23.880}, which, in the scenario of two parties, two settings, and two outcomes, states that all local $\vec{p}$ have to satisfy
\begin{equation}\label{eq: chsh}
\langle A_{1;0}A_{2;0}\rangle+ \langle A_{1;0}A_{2;1}\rangle + \langle A_{1;1}A_{2;0}\rangle -\langle A_{1;1}A_{2;1}\rangle\leqslant 2,
\end{equation}
where $A_{i;x_{i}}$ is a dichotomic observable with measurement outcomes $\pm 1$ describing the measurement of party $i$ corresponding to the measurement setting $x_{i}$, and $\langle A_{1;x_{1}}A_{2;x_{2}}\rangle$ are the standard expected values that express in terms of $p(\mathbf{a}|\mathbf{x})$. In our work, we consider only binary measurements, and so using expected values instead of probability distributions is more practical for our constructions. Moreover, the non-signalling constraints on the bi-local behaviours mean that the expected values encode all necessary information about the probability distribution. Therefore, in the rest of this work, we rely solely on the expected values formulation of Bell inequalities.

For any local behaviour, the expression in (\ref{eq: chsh}) is upper-bounded by $\beta_{L}=2$. In contrast, for quantum behaviours, the CHSH expression can attain values as high as  $\beta_{Q}=2\sqrt{2}$, referred to as the quantum bound or maximal quantum violation of the inequality. Noticeably, this maximal quantum value is attained by the maximally entangled state of two qubits, which is the simple example of a graph state. 

Building on the CHSH inequality, a construction of a class of multipartite Bell inequalities tailored to the $N$-qubit graph state was proposed later in \cite{PhysRevLett.124.020402}, which are sometimes referred to as the BASTA  inequalities. They can be stated in a compact form as
\begin{equation}\label{eq: graph_state_ineq}
|N_{r}|\langle \tilde{g}_{r}\rangle + \sum_{i\in [n]\setminus \{r\}} \langle \tilde{g}_{i}\rangle\leqslant |N_{r}|+n-1,
\end{equation}
where $r\in[n]$, and $\tilde{g}_{i}$ are operators constructed from the stabilizing operators $g_{i}$ of a given graph state $\ket{G}$ by the following substitution rules
\begin{equation}\label{eq: substitution_graph_state_ineq}
\begin{aligned}
X_{r}\rightarrow A_{r;0}+A_{r;1},\quad Z_{r}\rightarrow A_{r;0}-A_{r;1},\\
\forall_{i\in[n]\setminus \{r\}},\qquad X_{i}\rightarrow A_{i;0},\quad Z_{i}\rightarrow A_{i;1},
\end{aligned}
\end{equation}
where, as before, $A_{i;x_i}$ are binary Hermitian observables such that $A_{i;x_i}^2=\mathbbm{1}$ for all $x_u$ and $i$.
Throughout this paper, we refer to the party $r$ for which such a substitution is performed as a \textit{rotated party}. 

The maximal quantum value $\beta_{Q}=(2\sqrt{2}-1)|N_{r}|+n-1$ of the inequality \eqref{eq: graph_state_ineq} can be achieved by measuring the following observables
\begin{equation}\label{eq: observable_choice}
A_{r;j}=(X+ (-1)^{j}Z)/\sqrt{2},\quad A_{i;j}=X^{1-j}Z^{j},
\end{equation}
for $i\in[n]\setminus \{r\}$ and $j\in\{0,1\}$ on the graph state $\ket{G}$. 

To quantify the degree of violation for a given inequality, one uses the ratio $\beta_{Q}/\beta_L$, which for the inequality \eqref{eq: graph_state_ineq} amounts to 
\begin{equation}
\frac{\beta_{Q}}{\beta_{L}}=\frac{(2\sqrt{2}-1)|N_{r}|+n-1}{|N_{r}|+n-1}.
\end{equation}
As observed in Ref. \cite{PhysRevLett.124.020402}, for a given graph state, this ratio is maximised when $r$ is the vertex with the largest neighbourhood.

The ratio $\beta_{Q}/\beta_{L}$ is also a good metric for the noise robustness of the given inequality. Given a Bell functional $\langle I\rangle$ for which $\ket{\psi}$ achieves $\beta_{Q}$, if $\langle I\rangle=0$ for a maximally mixed state, then a state $\rho(p)= p \ket{\psi}\!\bra{\psi}+(1-p)2^{-n}\mathbb{1}$ violates inequality $\langle I \rangle \leqslant \beta_{L}$ if $p\geqslant \beta_{L}/\beta_{Q}$.

Notably, the BASTA inequalities were specifically constructed for detecting Bell nonlocality and, in particular, for self-testing graph states with a minimal number of expected values in mind. However, they are incapable of detecting GMNL behaviours originating from quantum states (as explained below). Arguably the most prominent Bell inequality for GMNL detection is that derived by Svetlichny \cite{PhysRevD.35.3066}, which we state here in a recursive form:
\begin{equation}\label{eq: svetlichny}
I_{n}=  I_{n-1}A_{n;0}+I'_{n-1}A_{n;1},
\end{equation}
where $I'_{n-1}$ is equal to $I_{n-1}$ up to the swap of measurement labels $0\leftrightarrow 1$ on each party, and $I_{2}$ is the CHSH operator. Then the inequality $\langle I_{n}\rangle\leq  2^{n-1}$ is satisfied for all bi-local behaviours, whereas the maximal quantum is $2^{n-1}\sqrt{2}$.

\section{Results}\label{sec:res}
Our aim is to adapt the BASTA inequality to detect genuine multipartite nonlocality, particularly in $N$-qubit graph states, while preserving their self-testing capabilities. Before introducing our construction in full generality, we begin with two illustrative examples; subsequently, we generalise these ideas to any graph. Lastly, we build on the same technique to derive the second class of inequalities detecting GMNL based on the Svetlichny inequality \cite{PhysRevD.35.3066}. 

\subsection{GMNL detection of a 3-qubit graph state}\label{sec:3partite}

Let us begin by showing that the BASTA inequalities are unable to detect GMNL in quantum systems. To this end, we consider the simplest tripartite graph $G=(\{1,2,3\},\{(1,2),(1,3)\})$ which corresponds to the $3$-qubit Greenberger-Horne-Zeilinger (GHZ) state \cite{Greenberger1989}. To construct an inequality maximally violated by the associated graph state $\ket{G}$, we have two choices of a rotated party: either $r=1$ or $r=2$, as the case $r=3$ is equivalent to the $r=2$ case due to the symmetry of $G$.

As per Eq. \eqref{eq: graph_state_ineq} the corresponding Bell functionals are
\begin{equation}\label{eq: G3_ineq_r=1}
\begin{aligned}
2 \langle (A_{1;0}+A_{1;1})A_{2;1} A_{3;1}\rangle + & \langle (A_{1;0}-A_{1;1})A_{2;0}\rangle \\ & +\langle (A_{1;0}-A_{1;1}) A_{3;0}\rangle,
\end{aligned}
\end{equation}
for $r=1$ and
\begin{equation}\label{eq: G3_ineq_r=2}
\begin{aligned}
\langle A_{1;0}(A_{2;0}-A_{2;1}) A_{3;1}\rangle + & \langle A_{1;1}(A_{2;0}+A_{2;1})\rangle\\
&+ \langle A_{1;1}A_{3;0}\rangle,
\end{aligned}
\end{equation}
for $r=2$.

To show that these functionals are incapable of detecting GMNL in quantum behaviors, let us show the existence of a bi-local behavior that gives rise to a higher value of these expressions that allowed by quantum theory, which are, $4\sqrt{2}$ and $2\sqrt{2}+1$ for Eq. \eqref{eq: G3_ineq_r=1} and \eqref{eq: G3_ineq_r=2}, respectively. For a bi-local strategy, let us consider a Popescu-Rohrlich (PR) box \cite{Popescu1994PRbox}, which is a non-signalling behaviour that achieves a value of $4$ on the CHSH expression. If we take a PR box between parties $1$ and $2$, and for party $3$ we choose a distribution such that $\langle A_{3;j}\rangle = 1$ for $j= 0,1$, then it is easy to see that both expressions \eqref{eq: G3_ineq_r=1} and \eqref{eq: G3_ineq_r=2} achieve the values of $6$ and $4$, respectively, both being larger than the maximal quantum values of these expressions.

Surprisingly, it is sufficient to modify the prefactor of a single expected value in both expressions to derive Bell inequalities that detect quantum GMNL. The modified Bell functionals read
\begin{equation}\label{eq: G3_gmnl_r=1}
\begin{aligned}
 \langle (A_{1;0}+A_{1;1})A_{2;1} A_{3;1}\rangle + & \langle (A_{1;0}-A_{1;1})A_{2;0}\rangle \\ & +\langle (A_{1;0}-A_{1;1}) A_{3;0}\rangle,
\end{aligned}
\end{equation}
for $r=1$ and
\begin{equation}\label{eq: G3_gmnl_r=2}
\begin{aligned}
\langle A_{1;0}(A_{2;0}-A_{2;1}) A_{3;1}\rangle + & \langle A_{1;1}(A_{2;0}+A_{2;1})\rangle\\
&+ 2 \langle A_{1;1}A_{3;0}\rangle,
\end{aligned}
\end{equation}
for $r=2$.

To compute the quantum bound for these functionals, let us first observe that Eq. \eqref{eq: G3_gmnl_r=2} is upper bounded by Eq. \eqref{eq: G3_ineq_r=2} plus $1$, meaning that the maximal quantum value of Eq. \eqref{eq: G3_gmnl_r=2} is upper bounded by $2\sqrt{2} +2$. Measuring observables \eqref{eq: observable_choice} on the graph state $\ket{G}$ yields the value $2\sqrt{2} +2$, proving that it is the quantum bound on Eq. \eqref{eq: G3_gmnl_r=2}.

As for the expression \eqref{eq: G3_gmnl_r=1}, finding its maximal quantum value requires an alternative approach based on the sum-of-squares decomposition. In fact, one directly checks that the following inequality 
holds true,
\begin{equation}
\begin{aligned}
2\sqrt{5}\mathbbm{1}-\mathcal{B}=&\frac{1}{\sqrt{5}}\left(\mathbb{1}-\sqrt{5}/2\;\tilde{g}_{1}\right)^{2}+\frac{2}{\sqrt{5}}\Big(\mathbb{1}-\sqrt{5}/4\;\tilde{g}_{2}\Big)^{2}\nonumber\\
&+\frac{2}{\sqrt{5}}\left(\mathbb{1}-\sqrt{5}/4\tilde{g}_{3}\right)^{2}.
\end{aligned}
\end{equation}
where $\mathcal{B}$ stands for the Bell operator corresponding to the expression \eqref{eq: G3_gmnl_r=1}. This simple sum-of-squares decomposition directly demonstrates that $2\sqrt{5}\mathbbm{1}-\mathcal{B}\geq 0$, implying that for quantum behaviours the largest value of the Bell functional \eqref{eq: G3_gmnl_r=1} is exactly $2\sqrt{5}$. Moreover, this value is achievable by taking the graph state $\ket{G}$ and the following observables 
\begin{equation}
\begin{aligned}
A_{1;j}=\frac{X+2(-1)^{j}Z}{\sqrt{5}},\qquad A_{i;0}=X,\qquad A_{i;1}=Z,
\end{aligned}
\end{equation}
for $i\in\{2,3\}$.

The maximal bi-local value $\beta_{BL}$ for both functionals can be readily computed by splitting the optimisation into three parts, one for each bipartition, and taking the maximal value over all of them. This approach takes advantage of the fact that for a linear functional $I(\cdot)$ and a bi-local distribution $p(\mathbf{a}|\mathbf{x})$, we have
\begin{equation}
I(p(\mathbf{a}|\mathbf{x}))= \sum_{\substack{Q\subset[n]\\ |Q|\leqslant \lfloor n/2\rfloor}}r_{Q}I(p_{Q|\overline{Q}}(\mathbf{a}|\mathbf{x})),
\end{equation}
where $p_{Q|\overline{Q}}(\mathbf{a}|\mathbf{x})$ is a behaviour local across $Q|\overline{Q}$ \eqref{eq:local_bip}. This implies that for any bi-local $p(\mathbf{a}|\mathbf{x})$ for which $I(p(\mathbf{a}|\mathbf{x}))=\beta_{BL}$, there exists a behaviour $p_{Q|\overline{Q}}(\mathbf{a}|\mathbf{x})$ for some $Q$, such that $I(p_{Q|\overline{Q}}(\mathbf{a}|\mathbf{x}))=\beta_{BL}$.

Let us begin with Eq. \eqref{eq: G3_gmnl_r=2}. The bipartition $2|13$ makes the party $2$ local, leaving us with two choices for the optimization: $\langle A_{2;0}+A_{2;1}\rangle=2$ or $\langle A_{2;0}-A_{2;1}\rangle=2$, which both lead to the value $4$.

Next, the optimization over distributions local with respect to the $1|23$ partition is equivalent to maximizing the following expression:
\begin{equation}
\begin{aligned}
\langle (A_{2;0}+A_{2;1})\rangle\langle A_{1;1}\rangle + \langle (A_{2;0}-A_{2;1})A_{3;1}\rangle\langle A_{1;0}&\rangle\\
+2\langle A_{1;1} \rangle  \langle A_{3;0}&\rangle. 
\end{aligned}
\end{equation}
Notice that we can upper-bound the last term by $2$ after which the first two terms are equivalent to a local model optimisation for the CHSH functional, and so the whole expression is upper-bounded by $4$.

Lastly, for the $12|3$ bipartition, we make use of the fact that
\begin{equation}
\langle A_{1;1}A_{2;j}\rangle + \langle A_{1;1}A_{3;0}\rangle \leqslant \langle A_{2;j}A_{3;0}\rangle +1,
\end{equation}
for $j=0,1$, which follows from the inequality derived in Ref. \cite{PhysRevA.72.022340}, which we generalise in the later part of this work (see Lemma \ref{lem: general_ineq_expected}). This allows us to upper-bound the Bell functional \eqref{eq: G3_gmnl_r=2} by
\begin{equation}
\langle A_{1;0} (A_{2;0}-A_{2;1}) \rangle \langle A_{3;1}\rangle + \langle (A_{2;0}+A_{2;1}) \rangle \langle A_{3;0}\rangle +2.
\end{equation}
We once again get a term related to a CHSH expression over a local model, which implies that the maximum of this expression is $4$.

As for the bi-local bound of the expression \eqref{eq: G3_gmnl_r=1}, the bipartition $1|23$ gives the value of $4$ by the same logic as for Eq. \eqref{eq: G3_gmnl_r=2} for the first bipartition. As for the $12|3$ and $13|2$ bipartitions, they can be computed simultaneously due to the symmetry of Eq. \eqref{eq: G3_gmnl_r=1}. For the $13|2$ bipartition, the sum of the first two terms is upper-bounded by $2$ for any distribution shared by parties $2$ and $3$. This yields a total upper bound of $4$ for the entire expression.

Since for every bipartition and both parties we always got the same value, we have that $\beta_{BL}=4$ for both Eq. \eqref{eq: G3_gmnl_r=1} and \eqref{eq: G3_gmnl_r=2}. This is smaller than the quantum bound for both of the functionals, showing that they are capable of detecting GMNL in quantum behaviours.

\subsection{GMNL inequalities for graph states}\label{sec:ineq_1}

The generalisation of the above tripartite inequalities to $n$-qubit graph states is surprisingly straightforward, as all the necessary components have already been discussed. First, to construct a GMNL witness centred on vertex $r$, we set the coefficients of all terms $\langle \tilde{g}_{i}\rangle$ with $i\in N_{r}$ to $1$, analogous to Eq. \eqref{eq: G3_gmnl_r=2}. For indices $i\notin N_{r}$, we upper-bound the expression using the inequality stated in the following lemma. 
\begin{lem}\label{lem: general_ineq_expected}
Let $S_{0}|S_{1}|S_{2}$ be a partition of the set of parties $[n]$ for $n\geqslant 3$ into three disjoint sets, and $A_{i;j}$ be a dichotomic observable with measurement outcomes $\pm 1$ acting on party $i$ with measurement setting $j$. Moreover, let us define collective observables 
\begin{equation}
O_{0;\vec{x}}\eqqcolon \prod_{i\in S_{0}}A_{i;x_{i}},\;\; O_{1;\vec{y}}\coloneqq  \prod_{i\in S_{1}}A_{i;y_{i}},\;\; O_{2;\vec{z}}\coloneqq \prod_{i\in S_{2}}A_{i;z_{i}}
\end{equation}
acting on spatially separated sets of parties $S_{0},S_{1},S_{2}$. Then for any choice of free parameters $\alpha_{0},\alpha_{1}\in [-1,1]$ and any non-signalling behaviour we have
\begin{equation}
\alpha_{0}\langle O_{0;\vec{x}}O_{1;\vec{y}}\rangle + \alpha_{1}\langle  O_{0;\vec{x}}O_{2;\vec{z}}\rangle \leqslant \alpha_{0}\alpha_{1}\langle O_{1;\vec{y}}O_{2;\vec{z}}\rangle + 1.
\end{equation}
\end{lem}
A proof can be found in Appendix \ref{app:ineq_basta}. As stated before, this is a direct generalisation of an inequality derived in Ref. \cite{PhysRevA.72.022340}. With this, we can proceed to the first main result of this work.

\begin{thm}\label{thm: gmnl_standard}
Let us consider a connected graph $G$ and associate a qubit with each vertex. We construct the following Bell functional associated with the graph $G$,
\begin{equation}\label{eq: gmnl_ineq_general}
\begin{aligned}
\langle I\rangle =& \sum_{i\in \dot{N}_{r}}\langle \tilde{g}_{i}\rangle+2\sum_{i\notin \dot{N}_{r}}\langle \tilde{g}_{i}\rangle,
\end{aligned}
\end{equation}
where $r\in [n]$ is the rotated party, $\dot{N}_{r}=N_{r}\cup\{r\}$, and $\tilde{g_{i}}$ is an operator constructed form a stabilizing operator $g_{i}$ of the graph state $\ket{G}$ via the substitution \eqref{eq: substitution_graph_state_ineq}. This functional is upper bounded by $\beta_{BL}=2(n-1)$ and $\beta_{Q}=2\sqrt{1+|N_{r}|^{2}} + 2(n-|N_{r}|-1)$ for bi-local and quantum behaviors respectively. 
\end{thm}

The proof uses the exact same ideas that were used for the case of $3$-vertex graph; the details can be found in Appendix~\ref{app:ineq_basta}. 

The intuition behind this inequality is straightforward. It comprises two parts, which we term the 'core' and the 'extension'. The core consists of party $r$ and its neighbours. It is constructed to ensure that certain quantum behaviours exceed the bounds of any behaviour which is local with respect to any arbitrary bipartition $Q|\overline{Q}$ such that $\dot{N}_{r}\not\subset Q,\overline{Q}$. The extension $i\notin \dot{N}_{r}$ is the part of a Bell functional that does not contribute directly to the advantage of quantum behaviours, but rather allows us to extend the inequality from the core to the rest of the parties, in such a way as to preserve the advantage of quantum behaviours over bi-local ones. This is performed by upper-bounding the original inequality by a new one using Lemma \ref{lem: general_ineq_expected}, similarly to how it was shown in the example in Sec. \ref{sec:3partite}.

The proof of the quantum bound $\beta_{Q}$ follows easily from a sum-of-squares decomposition (see Eq. \eqref{eq:sos_thm_1} in App. \ref{app:ineq_basta}), with the maximal quantum violation  achieved for a graph state $\ket{G}$ and the observables
\begin{equation}\label{eq: observable_choice_gmnl}
A_{r;k}=\frac{X+(-1)^{k}|N_{r}|Z}{\sqrt{1+|N_{r}|^{2}}},\qquad A_{i;j}=X^{1-j}Z^{j},
\end{equation}
for $i\in [n]\setminus \{r\}$.

A notable feature of our construction is that the maximal ratio $\beta_{Q}/\beta_{BL}$ is achieved by taking $r$ to be the vertex with minimal neighbourhood $|N_{r}|$. This stands in contrast to the BASTA inequality \eqref{eq: graph_state_ineq}, for which the maximal ratio $\beta_{Q}/\beta_{L}$ is reached by maximising $|N_{r}|$.

To illustrate how to apply Theorem \ref{thm: gmnl_standard} to a specific graph, let us consider another example $G=(V,E)$, where $V=[4]$ and $E={(1,2),(1,3),(2,3),(3,4)}$. As stated above, the optimal choice of the rotated party $r$ is the one minimising $|N_{r}|$. Considering the fact that $|N_{3}|=3$ and $|N_{1}|=|N_{2}|=2$, the optimal choice for this $G$ is $r=4$ with $|N_{r}|=1$. The generators of the graph state $\ket{G}$ are given by $g_{1}=X_{1}Z_{2}Z_{3}$, $g_{2}=Z_{1}X_{2}Z_{3}$, $g_{3}=Z_{1}Z_{2}X_{3}Z_{4}$, $g_{4}=Z_{3}X_{4}$. Then, following Theorem \ref{thm: gmnl_standard} we use substitution \eqref{eq: substitution_graph_state_ineq} to get $\tilde{g}_{1} = A_{1;0}A_{2;1} A_{3;1}$, $\tilde{g}_{2}=A_{1;1}A_{2;0}A_{3;1}$, $\tilde{g}_{3}=A_{1;1}A_{2;1}A_{3;0}(A_{4;0}-A_{4;1})$, and $\tilde{g}_{4}=A_{3;1}(A_{4;0}+A_{4;1})$. This gives us the following Bell inequality
\begin{equation}
\langle \tilde{g}_{3}\rangle + \langle \tilde{g}_{4}\rangle+2\langle\tilde{g}_{1}\rangle+2\langle\tilde{g}_{2}\rangle\leqslant 6,
\end{equation}
which is satisfied for all bi-local behaviours \eqref{eq:bilocal}. The quantum bound equals $\beta_{1}=2\sqrt{2}+4$ with the optimal state being $\ket{G}$ and the optimal observable choice given by Eq. \eqref{eq: observable_choice}.

Let us finally note that the Bell functional in Theorem \ref{thm: gmnl_standard} is also closely related to the Bell inequality derived in Ref. \cite{Coiteux-Roy2025-es} for a more restrictive notion of GMNL, called LOSR-GMNL. The advantage of our construction is that it is more versatile, i.e., it applies to all graph states, whereas the other construction does not, and it requires only two measurements per party.

\subsection{Local complementation improvement}\label{sec:LC}

Interestingly, the ratio $\beta_{Q}/\beta_{BL}$ can further be improved, beyond a simple optimization over $|N_{r}|$. The idea is to use local complementation on the target graph $G$ to find an equivalent graph $G'$ with a vertex of a smaller neighbourhood, which in turn increases $\beta_{Q}/\beta_{BL}$. To illustrate this idea, let us consider a complete graph $G$ consisting of four vertices, where all pairs of vertices are connected by an edge. Due to the symmetry, choosing any vertex as $r$ yields the ratio $\beta_{Q}/\beta_{BL}=\sqrt{10}/3\approx 1.05$. On the other hand, using local complementation, we can transform the fully connected graph $G$ into the star graph $G'=(\{1,2,3,4\},\{(1,2),(1,3),(1,4)\})$. Then, taking $r=2$, we get a ratio $\beta_{Q}/\beta_{BL}=(\sqrt{2}+2)/3\approx 1.14$.

This improvement through local complementation requires only a minor adjustment to our construction. Since $\ket{G}$ and $\ket{G'}$ are related by a local unitary transformation $U=U_1\otimes\ldots\otimes U_4$, adapting the inequality derived for $G'$ to the graph state $\ket{G}$ simply entails transforming the observables \eqref{eq: observable_choice_gmnl} by $U^{\dagger}$. For this example, the observables that give rise to the maximal violation of the inequality constructed for $G'$ with a graph state $\ket{G}$ are given by
\begin{equation}
A_{r;k}=\frac{X+(-1)^{k}Y}{\sqrt{2}},\qquad A_{i;j}=Y^{1-j}Z^{j},
\end{equation}
where $Y=\mathbb{i}XZ$ is the third Pauli matrix. 

While this strategy allows one to increase the ratio $\beta_{Q}/\beta_{BL}$ for the inequality originating from a graph $G$, finding the best graph $G'$ requires a search over the equivalence class of $G$ with respect to local complementation. The best known algorithm that finds the minimal degree up to local complementation runs in $O^{*}(1.938^{n})$ \cite{Cattan_o_2015}. Therefore, for large graphs, a complete search like this is inefficient. One approach to this issue is to make a "limited" search that does not guarantee the best possible ratio $\beta_{Q}/\beta_{BL}$, but yet has a chance of significantly improving it.

Alternatively, one can find an upper bound on $\beta_{Q}/\beta_{BL}$ without constructing $G'$ through an entanglement property called $k$-uniformity \cite{Gisin_1998, Higuchi_2000,K_obus_2019}. We say that a quantum state $\rho$ of $n$ qubits is $k$-uniform if all reduced states of $\rho$ after tracing out $n-k$ qubits equal the maximally mixed state. For graph states, this property allows us to lower bound the minimal size of any neighbourhood $N_{i}$ (and so also $|N_{r}|$) over all $G'$ by $|N_{i}|\geqslant k$. First, it is a well-known fact that for any graph state $\ket{G}$ we have $\ket{G}\!\bra{G}=1/2^{n}\prod_{i=1}^{n}(\mathbb{1}+g_{i})$. Then, since $\operatorname{Tr}_{[n]\setminus \dot{N}_{i}}(g_{i})\not\approx
\mathbb{1}$ for any $i\in[n]$, we have that $\operatorname{Tr}_{[n]\setminus \dot{N}_{i}}(\ket{G}\!\bra{G})\not\approx
\mathbb{1}$ for any $i\in[n]$. However, by definition, for a $k$-uniform state $\rho$, $\operatorname{Tr}_{S}(\rho)\approx \mathbb{1}$ for any $S\subset [n]$ s.t. $|S|=n-k$, and so $\operatorname{Tr}_{[n]\setminus \dot{N}_{i}}(\ket{G}\!\bra{G})\not\approx
\mathbb{1}$ implies $k\leqslant |N_{i}|$ for all $i\in G'$ and all $G'$. Substituting $N_{r}\rightarrow k$ in the expression for $\beta_{Q}$ gives us the following upper-bound
\begin{equation}
\frac{\beta_{Q}}{\beta_{BL}}\leqslant 1+\frac{\sqrt{1+k^{2}} -k}{n-1}.
\end{equation}
While computing $k$-uniformity is a computationally demanding task in and of itself, one can utilise this bound for some chosen classes of graph states, for which analytical computation of $k$-uniformity is possible.

\subsection{Self-testing}\label{sec:st}

Crucially, our construction retains the self-testing property of the underlying BASTA inequalities. To recall the task of self-testing, assume that a Bell experiment is performed with an $N$-partite quantum state $|\psi\rangle$ and quantum observables $A_{i,x_i}$, producing a behavior $\vec{p}$. Self-testing aims to determine the quantum state $|\psi\rangle$ and observables $A_{i,x_i}$ (up to certain equivalences) based on the observed nonlocal correlations $\vec{p}$. Since self-testing has already been established in Ref.  \cite{PhysRevLett.124.020402} for behaviours achieving the maximal quantum violation of the BASTA inequalities \eqref{eq: graph_state_ineq}, we employ a similar proof structure to show that the state and measurements maximizing the functional in Theorem \ref{thm: gmnl_standard} (achieving $\beta_{Q}$) can also be self-tested. 
\begin{thm}\label{thm: s-t}
Let $\ket{\psi_{\max}}\in \mathcal{H}=\bigotimes_{i=1}^{n}\mathcal{H}_{i}$ be an $n$-partite state and $A_{i;j}$ be finitely dimensional observables measured by party $i$ that give rise to the maximal quantum violation of the Bell functional \eqref{eq: gmnl_ineq_general} for a given connected graph $G$. Moreover, let us assume that a reduced state $\operatorname{Tr}_{[n]\setminus\{i\}}[\ket{\psi_{\max}}\!\bra{\psi_{\max}}]$ for all parties $i\in[n]$ is full-rank. Then, there exists a unitary operation $U=U_1\otimes\ldots\otimes U_n$ with $U_i$ being a local unitary such that $U_i:\mathcal{H}_i\to \mathbbm{C}^2\otimes\mathcal{H}_i'$ for some auxiliary Hilbert space $\mathcal{H}_i'$ and
\begin{equation}
\begin{aligned}
U\ket{\psi_{\max}}&=\ket{G}\ket{\mathrm{aux}},\\ U_i A_{i;j}\, U_i^{\dagger}&=X^{1-j}_{i}Z_{i}^{j}\otimes \mathbb{1}_{i}\quad \mathrm{for }\; i\in[n]\setminus\{r\},  \\
U_r A_{r;j}\,U_r^{\dagger}&=\frac{X_{r}+(-1)^{j}|N_{r}|Z_{r}}{\sqrt{1+|N_{r}|^{2}}}\otimes \mathbb{1}_{r},
\end{aligned}
\end{equation}
where $\ket{G}\in (\mathbb{C}^{2})^{\otimes n}$ is the graph state to be self-tested, $\ket{\mathrm{aux}}\in \bigotimes_{i=1}^{n}\mathcal{H}_{i}'$, and $\mathbb{1}_{i}$ is an identity operator acting on $\mathcal{H}_{i}'$.
\end{thm}
A proof can be found in Appendix \ref{app: s-t}. While Theorem \ref{thm: s-t} provides a self-testing statement, it requires that the measurement data achieve the maximal quantum value, which is impossible to realise in an experiment. To address this issue, we formulate and analyse a robust self-testing statement for a few exemplary graph states. 

Given an observed violation $\beta$ of a Bell inequality, a robust self-testing statement can be made by bounding the following quantity 
\begin{align}\label{eq:defF}
    F_{\rm DI}(\beta)=\inf_{{\rho}\in \mathcal{S}(\beta)} \max_{\lbrace{\Lambda_i} \rbrace_{i=1}^n}F((\Lambda_{1}\otimes\ldots \otimes \Lambda_{n})({\rho}),\ketbra{G}{G}),
\end{align}
where $\ket{G}$ is the target graph state, $\Lambda_i$ with $i\in [n]$ are quantum channels also called extraction channels, 
$F(\rho,\sigma)=(\Tr|\sqrt{\rho}\sqrt{\sigma}|)^2$ is the fidelity of quantum states, and the infimum taken over the set $\mathcal{S}(\beta)$ of all states that are compatible with the observed violation $\beta$ for the given inequality \cite{PhysRevA.80.062327,Kaniewski2015analytic}.
Operationally, Eq. \eqref{eq:defF} is the lowest fidelity with the self-tested state one can obtain 
by applying local extraction channels $\Lambda_i$, which in the self-testing context represent the local unitary operations $U_i$ and tracing out the auxiliary degrees of freedom, to the measured state $\rho$.

As an analysis of the noise-robustness of the family of inequalities \eqref{eq: gmnl_ineq_general}, we numerically compute lower bounds to the corresponding fidelity for two specific cases. In particular, we consider two families of inequalities detecting GMNL for the GHZ state, which corresponds to the $n-$partite star graph $G=(\{1,2,\ldots, n\},\{(1,2),(1,3), \ldots (1,n)\})$. In analogy with Sec. \ref{sec:3partite}, we consider both the case where the rotated party is $r = 1$ and $r = 2$. This leads to the following two families of inequalities
\begin{equation}\label{eq: NGHZ_gmnl_r=1}
\begin{aligned}
\langle (A_{1;0} &+A_{1;1})A_{2;1} \ldots A_{n;1}\rangle \\
& + \sum_{i = 2}^n \langle (A_{1;0} - A_{1;1}) A_{i;0}\rangle \leq 2 (n-1), 
\end{aligned}
\end{equation}
for $r=1$ and
\begin{equation}\label{eq: NGHZ_gmnl_r=2}
\begin{aligned}
\langle A_{1;1}(A_{2;0}+A_{2;1})& \rangle + \langle A_{1;0}(A_{2;0}-A_{2;1})A_{3;1} \ldots A_{n;1}\rangle\\
&+ \sum_{i =2}^n 2 \langle (A_{1;1} A_{i;0}\rangle \leq 2(n-1),
\end{aligned}
\end{equation}
for $r=2$.
The maximal quantum violation of the two inequalities amounts to $\beta_Q^{r=1} = 2 \sqrt{(n-1)^2 + 1}$ and $\beta_Q^{r=2} = 2 \sqrt{2} + 2 (n-2)$ respectively. Notice that for $n = 3$ the expressions \eqref{eq: NGHZ_gmnl_r=1} and \eqref{eq: NGHZ_gmnl_r=2} recover exactly those of the tripartite examples of Sec. \ref{sec:3partite}, namely \eqref{eq: G3_gmnl_r=1} and \eqref{eq: G3_gmnl_r=2} respectively. In Fig. \ref{fig:fidelity} we report the numerical lower bounds to the extractability \eqref{eq:defF} for the two families of inequalities above, for a number of particles $n = 3,\ldots,6$. Our numerical estimation, explained in detail in Appendix ~\ref{app:fid}, is based on the self-testing through operator inequalities methods of \cite{Kaniewski2015analytic,Coopmans2019robust}. The fidelity is compared to the one that can be extracted from the original BASTA inequalities, for the same choice of star graph and rotated parties. To make the two results comparable, the fidelity is plotted against the visibility of $p$, the violation, namely the amount of white noise that has not been mixed with the GHZ state to reach a given value of violation $\beta = p \beta_Q$. As expected, the GMNL Bell inequalities are more sensitive to noise and lead to a lower extractable fidelity for the same value of visibility.

\begin{figure*}
  \centering
  \includegraphics[width=0.88\linewidth]{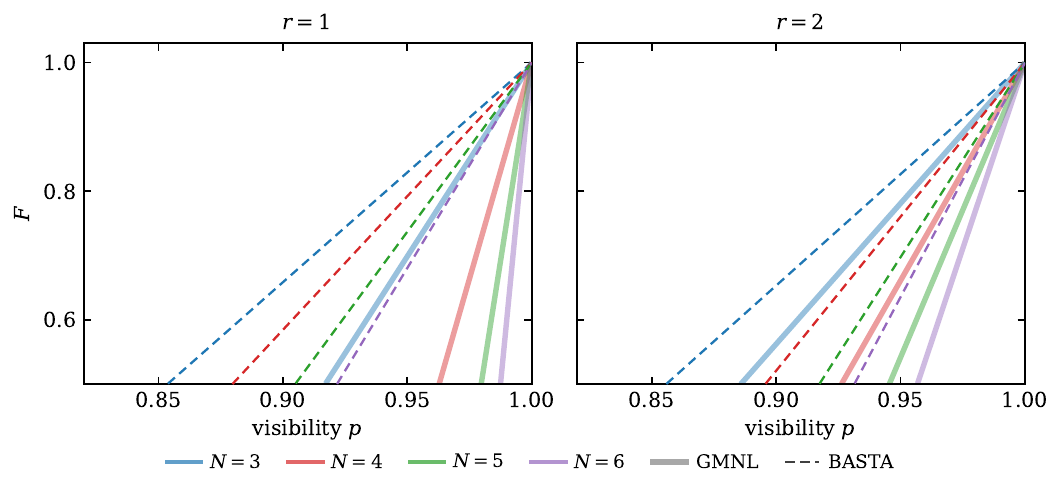}
  \caption{Numerical estimations of the fidelity with the target $N$-partite GHZ state as a function of the visibility $p = \beta/\beta_Q$, computed for both the inequalities where the rotated party is $r = 1$ and $r = 2$. The fidelity is the GMNL inequality is compared against the equivalent BASTA inequalities for the same graph and rotated party choice.}
  \label{fig:fidelity}
\end{figure*}

\subsection{Muli-rotation case}\label{sec:multi}
In Ref. \cite{PhysRevResearch.4.043215}, it was shown that the ratio $\beta_Q/\beta_L$ of the BASTA inequality can be improved by increasing the number of rotated parties. Consequently, a natural question arises regarding the potential for a similar improvement in the context of Bell inequalities stated in Theorem \ref{thm: gmnl_standard}.

Let us consider an example of a line graph $G=(\{1,2,3,4\},\{(1,2),(2,3),(3,4)\})$ and let us choose the rotated parties to be $r\in \{1,4\}$, giving us the following Bell functional
\begin{equation}
\begin{aligned}
\langle&(A_{1;0}+A_{1;1})A_{2;1}\rangle + \langle (A_{1;0}-A_{1;1})A_{2;0}A_{3;1}\rangle\\
&+ \langle A_{2;1}A_{3;0}(A_{4;0}-A_{4;1})\rangle + \langle A_{3;1}(A_{4;0}+A_{4;1})\rangle.
\end{aligned}
\end{equation}
It is direct to prove through a suitable sum-of-squares decomposition that the maximal quantum value of this expression amounts to $\beta_{Q}=4\sqrt{2}\approx 5.66$. However, a bi-local strategy employing a PR box for parties 3 and 4, combined with deterministic assignments $\langle A_{i;j}\rangle=1$ for $i=1,2$ and $j=0,1$, yields a value of 6. Consequently, this expression is incapable of detecting GMNL in quantum behaviours.

This is despite the fact that, for the same graph, choosing either $r=1$ or $r=4$ gives us a functional detecting GMNL. This shows another discrepancy between general Bell nonlocality and genuine multipartite nonlocality--an additional rotated party does not necessarily lead to an improvement in the ratio $\beta_Q/\beta_{BL}$. In fact, we suspect that in the case of GMNL detection, additional rotated parties always lead to a reduction of this ratio. In all examples we have examined, adding an extra rotated site left $\beta_{BL}$ unchanged, while decreasing $\beta_{Q}$.

\subsection{Clique-based improvement}\label{sec:clique}

The ratio of the inequality from Theorem \ref{thm: gmnl_standard} is limited by the performance of the expression $\sum_{i\in \dot{N}_{r}}\langle \tilde{g}_{i}\rangle$. If we could exchange it for a different functional that nonetheless allows for the application of Lemma \ref{lem: general_ineq_expected}, we could increase the ratio $\beta_{Q}/\beta_{BL}$. 

This can be achieved by increasing the number of correlators by applying Eq. \eqref{eq: substitution_graph_state_ineq} not only to generators, but also to some of their products. As it turns out, one very good strategy, that can be applied to any vertex $r$ for which its neighbourhood forms a clique, is to choose odd-numbered products of generators $g_{i}$ for which $i\in N_{r}$. Surprisingly, such a construction recovers a well-known Svetlichny inequality \cite{PhysRevD.35.3066}, which by itself allows for the detection of GMNL with a high ratio $\beta_{Q}/\beta_{BL}=\sqrt{2}$. Written in our notation, the Svetlichny functional \eqref{eq: svetlichny} can be rewritten, up to relabeling of measurements, as
\begin{equation}
\sum_{\substack{S\subseteq [m]\\|S| \textrm{- odd}}}\langle\tilde{f}_{S}\rangle,
\end{equation}
where we choose one of the vertices of a complete graph $G$ as the rotated party $r$, $\tilde{f}_{S}$ denotes the substitution \eqref{eq: observable_choice} acting on the reduced form of the operator, i.e., the resulting operator does not contain products of observables acting on the same party, and $f_{S}=\prod_{i\in S}g_{i}$ with $g_{i}$ defined in Eq. \eqref{eq: generators}. Notice that for $m=2$, this corresponds to the core in Eq. \eqref{eq: gmnl_ineq_general} for the case $|N_{r}|=1$, i.e., the case giving the largest ratio $\beta_{Q}/\beta_{BL}$. 

Before we move to the general construction of an inequality based on this core, let us consider the simplest example of such an inequality that differs from any inequality from Theorem \ref{thm: gmnl_standard}. To this end, let us consider a graph $G=(\{1,2,3,4\},\{(1,2),(1,3),(2,3),(3,4)\})$. This modified construction applies a Svetlichny inequality to the largest clique in a target graph - in this example, the largest clique consists of parties $1,2,3$. Then, we take $r=1$ and construct the following functional
\begin{equation}\label{eq:Svetlichny_example}
\begin{aligned}
\langle I\rangle=&\langle (A_{1;0}+A_{1;1})A_{2;1}A_{3;1}\rangle + \langle (A_{1;0}-A_{1;1})A_{2;0}A_{3;1}\rangle\\
&+
\langle (A_{1;0}-A_{1;1})A_{2;1}A_{3;0}A_{4;1}\rangle\\
&-\langle (A_{1;0}+A_{1;1})A_{2;0}A_{3;0}A_{4;1}\rangle +4\langle A_{3;1}A_{4;0}\rangle,
\end{aligned}
\end{equation}
in which the first three expected values correspond to the generators $g_{1},g_{2},g_{3}$, the fourth expected value corresponds to $f_{\{1,2,3\}}=g_{1}g_{2}g_{3}=-X_{1}X_{2}X_{3}Z_{4}$, and the last one originates from $g_{4}$. Comparing the first four terms to the Svetlichny inequality \cite{PhysRevD.35.3066}, we get that our core is equivalent to it up to the exchange of measurement settings in parties $2$ and $3$, and up to the additional observable $A_{4;1}$.

While we give a formal argument of it in the proof of the general construction, one can imagine that the presence of the Svetlichny-like core guarantees that we get a non-trivial bi-local bound $\beta_{BL}$ for any bipartition $Q|\overline{Q}$ that cuts through the core. The bi-local bound on the Svetlichny functional is $4$, and so, for all behaviours local across any such $Q|\overline{Q}$ \eqref{eq:local_bip} we can upper-bound Eq. \eqref{eq:Svetlichny_example} by $8$. Regarding the remaining $123\vert{}4$ bipartition, applying Lemma \ref{lem: general_ineq_expected} to the first two terms and the final term of Eq. \eqref{eq:Svetlichny_example} yields the upper bound: 
\begin{equation}
\begin{aligned}
\langle I\rangle\leqslant&\langle (A_{1;0}+A_{1;1})A_{2;1}A_{4;0}\rangle + \langle (A_{1;0}-A_{1;1})A_{2;0}A_{4;0}\rangle\\
&+
\langle (A_{1;0}-A_{1;1})A_{2;1}A_{3;0}A_{4;1}\rangle\\
&-\langle (A_{1;0}+A_{1;1})A_{2;0}A_{3;0}A_{4;1}\rangle +4,
\end{aligned}
\end{equation}
Since this takes the form of a Svetlichny-like inequality for parties 1, 2, and 4, the expression is bounded above by $8$ for any $123\vert{}4$-local behavior. Then, since a local, deterministic strategy assigning $+1$ to each observable achieves $8$, we get that $\beta_{BL}=8$.

As for the quantum bound, taking observables to be the ones maximizing the BASTA inequality, we get $\langle I\rangle = 4\sqrt{2}+4$. Moreover, the Svetlichny inequality is known to have a quantum upper bound of $4\sqrt{2}$, and so for our inequality $\beta_{Q}=4\sqrt{2}+4$.

Let us now properly introduce the general construction of GMNL detecting inequalities based on the Svetlichny core.
\begin{thm}\label{thm: clique}
Let $G$ be a graph such that there exists a vertex $r$ for which the induced subgraph of its neighbourhood $N_{r}$ forms a complete graph. Then, we construct the following functional associated with $G$
\begin{equation}\label{eq: gmnl_ineq_clique}
\sum_{\substack{S\subseteq \dot{N}_{r}\\ |S|\textrm{ - odd}}} \langle \tilde{f}_{S}\rangle +2^{|\dot{N}_{r}|-1} \sum_{k\in [n]\setminus \dot{N}_{r}}\langle\tilde{g}_{k}\rangle,
\end{equation}
where $\dot{N}_{r}=N_{r}\cup\{r\}$, and $f_{S}=\prod_{i\in S}g_{i}$ which under Eq. \eqref{eq: substitution_graph_state_ineq} yields
\begin{equation}
\tilde{f}_{S}= (-1)^{\frac{|S|-1}{2}}(A_{r;0}+\xi A_{r;1})\prod_{i\in S\setminus\{r\}}A_{i;0}\prod_{j\in \underset{k\in S}{\Delta}N_{k}\setminus\{r\}}A_{j;1}
\end{equation}
with $\xi=1$ if $r\in S$ and $\xi=-1$ for $r\notin S$. For this functional we have $\beta_{BL}=2^{|\dot{N}_{r}|-1}(n-|\dot{N}_{r}|+1)$ and $\beta_{Q}=2^{|\dot{N}_{r}|-1}(n-|\dot{N}_{r}|+\sqrt{2})$.
\end{thm}

\begin{figure}[!t]
    \centering
    \includegraphics[width=0.25\linewidth]{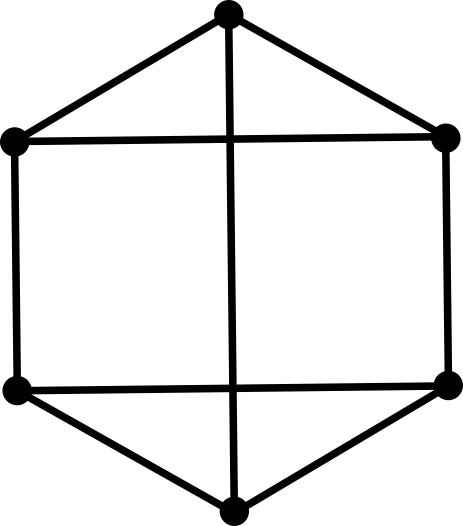}
    \caption{The only connected graph for up to $8$ vertices for which Theorem \ref{thm: clique} is inapplicable.}\label{fig: 6-graph}
\end{figure}

The proof can be found in Appendix \ref{app: clique}. It is not difficult to see that the ratio
\begin{equation}
\frac{\beta_{Q}}{\beta_{BL}}=\frac{n-|\dot{N}_{r}|+\sqrt{2}}{n-|\dot{N}_{r}|+1},
\end{equation}
is always larger than or equal to the ratio for the functional in Theorem \ref{thm: gmnl_standard}. However, this comes with two caveats. First, the number of correlators in the functional \eqref{eq: gmnl_ineq_clique} grows exponentially with the size of $|\dot{N}_{r}|$. Second, the functional \eqref{eq: gmnl_ineq_clique} is not directly applicable to all connected graphs. Luckily, as was the case for optimising ratio for \eqref{eq: gmnl_ineq_general}, here we also can make use of local complementations to identify different functionals for the same graph state $\ket{G}$.

Using the classification of local complementation equivalence classes of graphs up to $8$ vertices from Ref. \cite{PhysRevA.69.062311,Cabello_2009}, we show in Appendix \ref{app: lc} that out of 146 classes, there are 141 classes for which Theorem \ref{thm: clique} not only can be utilized, but it also offers an advantage in the $\beta_{Q}/\beta_{BL}$ ratio over Theorem \ref{thm: gmnl_standard}. Moreover, out of the remaining $5$ classes, Theorem \ref{thm: clique} can still be applied to $4$ of them, but in these cases Theorem \ref{thm: clique} and Theorem \ref{thm: gmnl_standard} produce the same functional. 

Therefore, there is only one class of graph states of up to $8$ qubits, for which the construction in Theorem \ref{thm: clique} is inapplicable - a $6$-vertex graph shown in Fig. \ref{fig: 6-graph}. Interestingly, out of the $4$ graph classes for which the two functionals are equivalent, one is the trivial $2$-vertex graph, whereas the other $3$ are directly related to the graph in Fig. \ref{fig: 6-graph}. 

Lastly, let us note that the inequality from Theorem \ref{thm: clique} can also be used for self-testing. However, in this case the self-testing strategy of the BASTA inequality applies directly. This is because all correlators of the BASTA inequality appear as well in Theorem \ref{thm: clique}, and because the optimal strategy for quantum violation is exactly the same between the two inequalities. For this reason we do not provide the self-testing statement here, instead directing the reader to Ref. \cite{PhysRevLett.124.020402}.

\section{Conclusions}\label{sec:con}

We introduce a family of Bell inequalities designed to detect genuine multipartite nonlocality (GMNL) in quantum correlations. These inequalities are specifically tailored to multiqubit graph states, as they attain their maximal quantum violation for this class of states. Moreover, this maximal violation allows for the self-testing of graph states, providing device-independent certification of the underlying quantum state. 

There are three notable benefits of our construction that make it particularly appealing from an experimental point of view. First, it requires only two measurement settings per party, which is in fact the minimal number necessary to observe quantum nonlocality. Second, for any given graph state, the construction yields a wide family of inequalities from which one can choose. This flexibility allows for an adjustable trade-off between the number of correlators that must be estimated and the magnitude of the ratio between the quantum value and the bi-local bound of the inequality. Increasing this ratio directly enhances the robustness of the inequality against white noise acting on the state. Third, all of the proposed inequalities detect genuine multipartite nonlocality, one of the strongest variants of multipartite nonlocality.

This work leaves many avenues for future investigation. Most importantly, the construction of inequalities that we showed for two specific cases could be extended to other scenarios. For example, using one of the $d$-outcome, bipartite inequalities \cite{PhysRevLett.88.040404, PhysRevLett.119.040402, Kaniewski2019maximalnonlocality, meyer2025quditclauserhorneshimonyholtinequalitynonlocality} could potentially allow one to tailor such inequalities for broader classes of states, like qudit graph states. Additionally, one could also apply the same technique for inequality detecting an even stronger version of multipartite nonlocality called full nonlocality \cite{Almeida_2010}. As inequalities tailored for this task are famously difficult to formulate, a construction such as the one in this work would be highly beneficial.

Another avenue for further exploration are the inequalities with multiple rotated parties. Even though we have provided an example for which such a construction fails in the sense that an additional rotation does not increase the ratio $\beta_Q/\beta_{C}$, we do not have a proof that introducing more rotated sites does not offer benefits for some other graph structures. This aspect certainly requires further investigation. 

Lastly, the applicability of Theorem \ref{thm: clique} remains unclear. While we have shown that this construction fails only for one out of the $146$ graph states of up to $8$ vertices, we are unsure if other such examples do not exist for larger graphs.

\begin{acknowledgments}
We thank Jordi Tura and Edwin Lobo for useful discussions. The Center for Quantum Enabled Computing project is carried out
within the International Research Agendas programme
of the Foundation for Polish Science co-financed by the
European Union under the European Funds for Smart
Economy 2021-2027 (FENG).
F.B. acknowledges financial support from the European Union’s Horizon Europe research and innovation programme under the Marie Skłodowska-Curie Action for the project No. 101148556 (ENCHANT).
The views and opinions expressed here are solely those of the authors and do not necessarily reflect those of the funding institutions.
The funding institutions cannot be held responsible for them.
\end{acknowledgments}
\section{AI usage}
Several LLMs have been used for error detection in the manuscript at the last stage of the writing process. All results have been derived and verified by the authors, and the manuscript has been fully written by the authors.

\bibliographystyle{apsrev4-1}
\bibliography{bib}

\onecolumngrid
\appendix

\section{Proof of Theorem \ref{thm: gmnl_standard}}\label{app:ineq_basta}

Here we prove Theorem \ref{thm: gmnl_standard}. To this end, we start by proving Lemma \ref{lem: general_ineq_expected}.
\setcounter{lem}{0}
\begin{lem}
Let $S_{0}|S_{1}|S_{2}$ be a partition of the set of parties $[n]$ for $n\geqslant 3$ into three disjoint sets, and $A_{i;j}$ be a dichotomic observable with measurement outcomes $\pm 1$ acting on party $i$ with measurement setting $j$. Moreover, let us define collective observables $O_{0;\vec{x}}\eqqcolon \prod_{i\in S_{0}}A_{i;x_{i}},O_{1;\vec{y}}\coloneqq  \prod_{i\in S_{1}}A_{i;y_{i}},O_{2;\vec{z}}\coloneqq \prod_{i\in S_{2}}A_{i;z_{i}}$ acting on spatially separated sets of parties $S_{0},S_{1},S_{2}$. Then for any choice of free parameters $\alpha_{0},\alpha_{1}\in [-1,1]$ and any non-signalling behaviour we have
\begin{equation}
\alpha_{0}\langle O_{0;\vec{x}}O_{1;\vec{y}}\rangle + \alpha_{1}\langle  O_{0;\vec{x}}O_{2;\vec{z}}\rangle \leqslant \alpha_{0}\alpha_{1}\langle O_{1;\vec{y}}O_{2;\vec{z}}\rangle + 1.
\end{equation}
\end{lem}
\begin{proof}
Let us begin with the following expression
\begin{equation}
\alpha_{0}p(O_{0;\vec{x}}=O_{1;\vec{y}})+\alpha_{1}p(O_{0;\vec{x}}=O_{2;\vec{z}})-\alpha_{0}\alpha_{1}p(O_{1;\vec{y}}=O_{2;\vec{z}}),
\end{equation}
where $p(X=Y)$ is the probability of getting the same outcome of measurement of $X$ as of $Y$. Written explicitly in the form $p(o_{0},o_{1},o_{2}|\vec{x},\vec{y},\vec{z})$, where $o_{0},o_{1},o_{2}$ are the outcomes of $O_{0;\vec{x}},O_{1;\vec{y}},O_{2;\vec{z}}$, the above expression equals
\begin{equation}
\begin{aligned}
\alpha_{0}p(O_{0;\vec{x}}=O_{1;\vec{y}})+\alpha_{1}p(O_{0;\vec{x}}=O_{2;\vec{z}})-\alpha_{0}\alpha_{1}p(O_{1;\vec{y}}=O_{2;\vec{z}})=&(\alpha_{0}+\alpha_{1}-\alpha_{0}\alpha_{1})(p(1,1,1|\vec{x},\vec{y},\vec{z})\\
&+p(-1,-1,-1|\vec{x},\vec{y},\vec{z})) \\
&+\alpha_{0}(p(1,1,-1|\vec{x},\vec{y},\vec{z})+p(-1,-1,1|\vec{x},\vec{y},\vec{z}))\\
&+\alpha_{1}(p(1,-1,1|\vec{x},\vec{y},\vec{z})+p(-1,1,-1|\vec{x},\vec{y},\vec{z}))\\
&-\alpha_{0}\alpha_{1}(p(-1,1,1|\vec{x},\vec{y},\vec{z})+p(1,-1,-1|\vec{x},\vec{y},\vec{z})).
\end{aligned}
\end{equation}
Since we have $\sum_{o_{0},o_{1},o_{2}}p(o_{0},o_{1},o_{2}|\vec{x},\vec{y},\vec{z})=1$, we can upper bound the right-hand side by $\max(\alpha_{0}+\alpha_{1}-\alpha_{0}\alpha_{1},\alpha_{0},\alpha_{1},-\alpha_{0}\alpha_{1})$. Then using non-signalling assumption we can make a substitution $p(X=Y)=(1+\langle XY\rangle)/2$ to the left-hand side we get
\begin{equation}
\alpha_{0}\langle O_{0;\vec{x}}O_{1;\vec{y}}\rangle +\alpha_{1}\langle O_{0;\vec{x}}O_{2;\vec{z}}\rangle -\alpha_{0}\alpha_{1}\langle O_{1;\vec{y}}O_{2;\vec{z}}\rangle \leqslant \max(\alpha_{0}+\alpha_{1}-\alpha_{0}\alpha_{1}, \alpha_{0}\alpha_{1}+\alpha_{0}-\alpha_{1},\alpha_{0}\alpha_{1}+\alpha_{1}-\alpha_{0},-\alpha_{0}-\alpha_{1}-\alpha_{0}\alpha_{1}).
\end{equation}
Lastly, note that we can rewrite the individual expressions in the $\max(\cdot)$ as
\begin{equation}
\begin{aligned}
\alpha_{i}\alpha_{j}+\alpha_{i}-\alpha_{j}&=-(1-\alpha_{i})(1+\alpha_{j})+1, \quad \textrm{for }i,j\in \{0,1\},\\
\alpha_{0}+\alpha_{1}-\alpha_{0}\alpha_{1}&= -(1-\alpha_{0})(1-\alpha_{1})+1,\\
-\alpha_{0}-\alpha_{1}-\alpha_{0}\alpha_{1}&= -(1+\alpha_{0})(1+\alpha_{1})+1,
\end{aligned}
\end{equation}
and since $(1+ \alpha_{0})(1+ \alpha_{1})$ is non-negative for $\alpha_{0},\alpha_{1}\in[-1,1]$, we have
\begin{equation}
\alpha_{0}\langle O_{0;\vec{x}}O_{1;\vec{y}}\rangle +\alpha_{1}\langle O_{0;\vec{x}}O_{2;\vec{z}}\rangle -\alpha_{0}\alpha_{1}\langle O_{1;\vec{y}}O_{2;\vec{z}}\rangle \leqslant 1.
\end{equation}
\end{proof}

Let us note here that the restriction to non-signalling sets is not necessary if one conditions the expected value by the choice of measurements of the party that act "trivially", e.g., $\langle O_{0;\vec{x}}O_{1;\vec{y}}\rangle|_{\vec{z}}$. This, however, will not be needed for our applications, and so we formulated it in terms of non-signalling correlations as it simplifies the notation. 

To streamline the usage of Lemma \ref{lem: general_ineq_expected} in the proof of Theorem \ref{thm: gmnl_standard} we formulate and prove the following corollary.
\begin{cor}\label{cor:bound_2}
Let $D_{0;0},D_{0;1},D_{1;0},D_{1;1},D_{2;0},D_{2;1}$ be disjoint subsets of $[n]$, and let $r$ be a part such that $r\in [n]$ but $r\notin D_{i;j}$ for any $i,j$. Then we have
\begin{equation}
\begin{aligned}
\alpha_{0}&\left\langle(A_{r;0}+A_{r;1})\prod_{i_{0}\in D_{0;0}}A_{i_{0};0}\prod_{i_{1}\in D_{0;1}}A_{i_{1};1}\prod_{i_{2}\in D_{1;0}}A_{i_{2};0}\prod_{i_{3}\in D_{1;1}}A_{i_{3};1}\right\rangle\\
&+ \alpha_{1}\left\langle (A_{r;0}-A_{r;1})\prod_{i_{0}\in D_{0;0}}A_{i_{0};0}\prod_{i_{1}\in D_{0;1}}A_{i_{1};1}\prod_{i_{2}\in D_{2;0}}A_{i_{2};0}\prod_{i_{3}\in D_{2;1}}A_{i_{3};1}\right\rangle \leqslant 2,
\end{aligned}
\end{equation}
for any $\alpha_{0},\alpha_{1}\in[-1,1]$.
\end{cor}
\begin{proof}
The proof follows from a double application of Lemma \ref{lem: general_ineq_expected}. First, let us take
\begin{equation}
O_{0,\vec{x}}=A_{r;0}\prod_{i_{0}\in D_{0;0}}A_{i_{0};0}\prod_{i_{1}\in D_{0;1}}A_{i_{1};1}, \quad O_{1;\vec{y}}=\prod_{i_{2}\in D_{1;0}}A_{i_{2};0}\prod_{i_{3}\in D_{1;1}}A_{i_{3};1}, \quad O_{2;\vec{z}}=\prod_{i_{2}\in D_{2;0}}A_{i_{2};0}\prod_{i_{3}\in D_{2;1}}A_{i_{3};1}.
\end{equation}
Since the sets $D_{i;j}$ are disjoint, we can use Lemma \ref{lem: general_ineq_expected} which yields 
\begin{equation}
\begin{aligned}
\alpha_{0}&\left\langle A_{r;0}\prod_{i_{0}\in D_{0;0}}A_{i_{0};0}\prod_{i_{1}\in D_{0;1}}A_{i_{1};1}\prod_{i_{2}\in D_{1;0}}A_{i_{2};0}\prod_{i_{3}\in D_{1;1}}A_{i_{3};1}\right\rangle\\
&+ \alpha_{1}\left\langle A_{r;0}\prod_{i_{0}\in D_{0;0}}A_{i_{0};0}\prod_{i_{1}\in D_{0;1}}A_{i_{1};1}\prod_{i_{2}\in D_{2;0}}A_{i_{2};0}\prod_{i_{3}\in D_{2;1}}A_{i_{3};1}\right\rangle \\
\leqslant& \alpha_{0}\alpha_{1}\left\langle\prod_{i_{0}\in D_{1;0}}A_{i_{0};0}\prod_{i_{1}\in D_{1;1}}A_{i_{1};1}\prod_{i_{2}\in D_{2;0}}A_{i_{2};0}\prod_{i_{3}\in D_{2;1}}A_{i_{3};1}\right\rangle+1.
\end{aligned}
\end{equation}
Next, we apply Lemma \ref{lem: general_ineq_expected} for the second time by taking 
\begin{equation}
O_{0,\vec{x}}=A_{r;1}\prod_{i_{0}\in D_{0;0}}A_{i_{0};0}\prod_{i_{1}\in D_{0;1}}A_{i_{1};1}, \quad O_{1;\vec{y}}=\prod_{i_{2}\in D_{1;0}}A_{i_{2};0}\prod_{i_{3}\in D_{1;1}}A_{i_{3};1}, \quad O_{2;\vec{z}}=-\prod_{i_{2}\in D_{2;0}}A_{i_{2};0}\prod_{i_{3}\in D_{2;1}}A_{i_{3};1},
\end{equation}
which gives us 
\begin{equation}
\begin{aligned}
\alpha_{0}&\left\langle A_{r;1}\prod_{i_{0}\in D_{0;0}}A_{i_{0};0}\prod_{i_{1}\in D_{0;1}}A_{i_{1};1}\prod_{i_{2}\in D_{1;0}}A_{i_{2};0}\prod_{i_{3}\in D_{1;1}}A_{i_{3};1}\right\rangle\\
&- \alpha_{1}\left\langle A_{r;1}\prod_{i_{0}\in D_{0;0}}A_{i_{0};0}\prod_{i_{1}\in D_{0;1}}A_{i_{1};1}\prod_{i_{2}\in D_{2;0}}A_{i_{2};0}\prod_{i_{3}\in D_{2;1}}A_{i_{3};1}\right\rangle \\
\leqslant& -\alpha_{0}\alpha_{1}\left\langle\prod_{i_{0}\in D_{1;0}}A_{i_{0};0}\prod_{i_{1}\in D_{1;1}}A_{i_{1};1}\prod_{i_{2}\in D_{2;0}}A_{i_{2};0}\prod_{i_{3}\in D_{2;1}}A_{i_{3};1}\right\rangle+1.
\end{aligned}
\end{equation}
Adding the two bounds together completes the proof.
\end{proof}

\setcounter{thm}{0}
\begin{thm}
Let us consider a connected graph $G$ and associate a qubit with each vertex of $G$. 
Let us then construct the following Bell functional associated with the graph $G$:
\begin{equation}\label{eq: gmnl_ineq_general_app}
\begin{aligned}
\langle I\rangle =& \sum_{i\in \dot{N}_{r}}\langle \tilde{g}_{i}\rangle+2\sum_{i\notin \dot{N}_{r}}\langle \tilde{g}_{i}\rangle,
\end{aligned}
\end{equation}
where $r\in [n]$, $\dot{N}_{r}=N_{r}\cup\{r\}$, and $\tilde{g_{i}}$ is an operator constructed from a stabilizing operator $g_{i}$ of the graph state $\ket{G}$ via the substitution \eqref{eq: substitution_graph_state_ineq}. Then, the maximal values over bi-local and quantum correlations are $\beta_{BL}=2(n-1)$ and $\beta_{Q}=2\sqrt{1+|N_{r}|^{2}} + 2(n-|N_{r}|-1)$, respectively. 
\end{thm}

\begin{proof}
We start the proof from showing, that the bi-local bound for functional \eqref{eq: gmnl_ineq_general_app} equals $\beta_{BL}=2(n-1)$. From the linearity of the Bell functional, it follows that to prove that this bound holds for all bi-local behaviours, it is sufficient to show that this inequality holds for behaviours local across $Q|\overline{Q}$, for all bipartitions $Q|\overline{Q}$. In fact, we can take it one step further and discard the dependence on $\lambda$, since the maximum of a linear Bell functional has to be achieved for some $\lambda=\lambda_{\max}$. Putting it all together, we optimise $\langle I\rangle$ over all bipartitions $Q|\overline{Q}$, for behaviours satisfying
\begin{equation}\label{eq:product_behaviour}
p(\mathbf{a}|\mathbf{x})=p(\mathbf{a}_{Q}|\mathbf{x}_{Q})p(\mathbf{a}_{\overline{Q}}|\mathbf{x}_{\overline{Q}}).
\end{equation}

To simplify the analysis, without loss of generality, we assume that $r\in Q$. Moreover, let us introduce the following notation $g_{i}=g_{i}^{Q}\otimes g_{i}^{\overline{Q}}$, i.e., $g_{i}^{Q}$ is the part of the Pauli string $g_{i}$ that acts on the set of qubits $Q$. Then, for any behaviour \eqref{eq:product_behaviour} we have
\begin{equation}\label{eq:I_expression}
\begin{aligned}
\langle I \rangle =  \sum_{i\in\dot{N}_{r}}\left\langle \tilde{g}_{i}^{Q}\right\rangle_{Q}\left\langle \tilde{g}_{i}^{\overline{Q}}\right\rangle_{\overline{Q}}+2\sum_{i\notin\dot{N}_{r}}\left\langle \tilde{g}_{i}^{Q}\right\rangle_{Q}\left\langle \tilde{g}_{i}^{\overline{Q}}\right\rangle_{\overline{Q}},
\end{aligned}
\end{equation}
where $\langle \cdot \rangle_{Q}$ is evaluated over $p(\mathbf{a}_{Q}|\mathbf{x}_{Q})$ and likewise for $\langle \cdot \rangle_{\overline{Q}}$. 

We split the rest of the proof into two cases: $\dot{N}_{r}\cap Q \neq \dot{N}_{r}$ and $\dot{N}_{r}\cap Q = \dot{N}_{r}$, which corresponds to the scenario in which $Q|\overline{Q}$ respectively cuts and does not cut across the core of the inequality. The plan is as follows: we will first show how a bipartition $Q|\overline{Q}$ cutting across the core of the inequality causes a non-trivial bound for the behaviour local across $Q|\overline{Q}$. Then, for the case in which $Q|\overline{Q}$ does not cut across the core, we will show how to "bring" a party outside of the core into it by using the inequality extension and Lemma \ref{lem: general_ineq_expected}.

We begin with the former case. We will reduce proving the bound to an almost bipartite problem between $r$ (which is in $Q$ by definition) and a party $l$ that is both in $\overline{Q}$ and in $N_r$. Notice that such party $l$ has to exist because $\dot{N}_{r}\cap Q \neq \dot{N}_{r}$. Due to the substitution rules for the rotated party $r$, we have that $\langle \tilde{g}_{i}^{Q}\rangle_{Q}\leqslant 2$ if $i\in \dot{N}_{r}$, whereas $\langle \tilde{g}_{i}^{Q}\rangle_{Q}\leqslant 1$ for $i\notin \dot{N}_{r}$ as well as $\langle \tilde{g}_{i}^{\overline{Q}}\rangle_{\overline{Q}}\leqslant 1$ for any $i$. Let us make use of this fact for all $i\neq r,l$. There are $n-|\dot{N}_{r}|$ terms $i\notin\dot{N}_{r}$ in Eq. \eqref{eq:I_expression}, each upper-bounded by $1$, and since every term is multiplied by $2$ we get a collective upper-bound of $2(n-|\dot{N}_{r}|)$. Next, there are $|\dot{N}_{r}|-2$ terms in $\dot{N}_{r}\setminus\{r,l\}$, each upper-bounded by $2$. Then we have

\begin{equation}
\begin{aligned}
\langle I \rangle \leqslant&\Bigg( \left\langle \tilde{g}_{r}^{Q}\right\rangle_{Q}\left\langle \tilde{g}_{r}^{\overline{Q}}\right\rangle_{\overline{Q}}+ \left\langle \tilde{g}_{l}^{Q} \right\rangle_{Q}\left\langle \tilde{g}_{l}^{\overline{Q}}\right\rangle_{\overline{Q}}\Bigg)+ 2(n-2)\\
=&\left( \left\langle (A_{r;0}+A_{r;1})\prod_{i\in N_{r}\cap Q}A_{i;1}\right\rangle_{Q}\left\langle\prod_{i\in N_{r}\cap \overline{Q}}A_{i;1}\right\rangle_{\overline{Q}}+ \left\langle (A_{r;0}-A_{r;1})\prod_{i\in N_{l}\cap Q\setminus\{r\}}A_{i;1}\right\rangle_{Q}\left\langle A_{l;0}\prod_{i\in N_{l}\cap\overline{Q}}A_{i;1}\right\rangle_{\overline{Q}}\right)\\
&+ 2(n-2).
\end{aligned}
\end{equation}
Finally, to bound the remaining expected values, we need to make use of Corollary \ref{cor:bound_2}. First, by taking $\alpha_{0}=\left\langle\prod_{i\in N_{r}\cap \overline{Q}}A_{i;1}\right\rangle_{\overline{Q}}$, $\alpha_{1}=\left\langle A_{l;0}\prod_{i\in N_{l}\cap\overline{Q}}A_{i;1}\right\rangle_{\overline{Q}}$ we get
\begin{equation}
\begin{aligned}
\langle I \rangle \leqslant
&\left( \alpha_0 \left\langle (A_{r;0}+A_{r;1})\prod_{i\in N_{r}\cap Q}A_{i;1}\right\rangle_{Q} + \alpha_1 \left\langle (A_{r;0}-A_{r;1})\prod_{i\in N_{l}\cap Q\setminus\{r\}}A_{i;1}\right\rangle_{Q} \right)+ 2(n-2).
\end{aligned}
\end{equation}
Then, we take 
\begin{equation}
D_{0;1}=N_{r}\cap N_{l}\cap Q,\quad D_{1;1}=(N_{r}\setminus N_{l})\cap Q,\quad D_{2;1}=((N_{l}\setminus N_{r})\cap Q)\setminus\{r\},
\end{equation}
and $D_{i;0}=\emptyset$ for all $i\in\{0,1,2\}$. Since these sets are disjoint, the conditions of Corollary \ref{cor:bound_2} are satisfied, giving us
\begin{equation}
\begin{aligned}
\langle I \rangle \leqslant
&\left( \alpha_0 \left\langle (A_{r;0}+A_{r;1})\prod_{i\in N_{r}\cap Q}A_{i;1}\right\rangle_{Q} + \alpha_1 \left\langle (A_{r;0}-A_{r;1})\prod_{i\in N_{l}\cap Q\setminus\{r\}}A_{i;1}\right\rangle_{Q} \right)+ 2(n-2)\\
\leqslant& 2(n-1)
\end{aligned}
\end{equation}

Now we tackle the case $\dot{N}_{r}\cap Q = \dot{N}_{r}$. Let us consider a party $l\in \overline{Q}$. Clearly, $l\notin N_{r}$, however, since we assume that $G$ is connected, there exists a path from $r$ to $l$. Let us denote by $t_{0},t_{1},\ldots,t_{m}$ the parties that form the shortest path between $r$ and $l$, where $t_{0}=r$ and $t_{m}=l$. 

For a given path $t_{0},t_{1},\ldots,t_{m}$, we define two sets $T_{0}\coloneq \{t_{2u}\}_{u=1}^{\lfloor m/2\rfloor}$, $T_{1}\coloneq\{t_{2u+1}\}_{u=1}^{\lceil m/2 \rceil-1}$, where $\lceil\cdot\rceil$, $\lfloor\cdot\rfloor$ are the ceiling and floor functions. The purpose of these sets is as follows: we want to create a CHSH-like term between parties $r$ and $l$. This will be achieved with Lemma \ref{lem: general_ineq_expected}, where correlators corresponding to $T_{0}$ will be responsible for creating correlators between  $A_{l;m'}$ and $A_{r;0}+A_{r;1}$ with $T_{1}$ creating the correlators between $A_{l;m'+1}$ and $A_{r;0}-A_{r;1}$, where $m'=m\mod 2$. With this idea in mind, we split $\langle I \rangle$ into three parts
\begin{equation}\label{eq:three_parts}
\langle I\rangle = \left(\langle\tilde{g}_{r}\rangle+2\sum_{i\in T_{0}}\langle\tilde{g}_{i}\rangle\right)+\left(\langle \tilde{g}_{t_{1}}\rangle+2\sum_{i\in T_{1}}\langle \tilde{g}_{i}\rangle\right)+\left(\sum_{i\in N_{r}\setminus\{t_{1}\}}\langle \tilde{g}_{i}\rangle+2\sum_{i\notin \dot{N}_{r}\cup T_{0}\cup T_{1}}\langle \tilde{g}_{i}\rangle\right).
\end{equation}
The purpose of the three parts is to upper-bound the above expression by another that contains an explicit CHSH-like term between parties $r$ and $l$, and so explicitly shows the upper-bound on behaviours local across $Q|\overline{Q}$. To this end, we again make use of Lemma \ref{lem: general_ineq_expected}.

Let us focus on the first term
\begin{equation}\label{eq:first_term}
\langle\tilde{g}_{r}\rangle+2\sum_{i\in T_{0}}\langle\tilde{g}_{i}\rangle=\sum_{k=0}^{1}\left\langle A_{r;k}\prod_{i\in N_{r}}A_{i;1} \right\rangle + 2\sum_{u=1}^{\lfloor m/2\rfloor}\left\langle A_{t_{2u};0}\prod_{j\in N_{t_{2u}}}A_{j;1}\right\rangle.
\end{equation}
We can make use of Lemma \ref{lem: general_ineq_expected} by taking $\alpha_{0},\alpha_{1}=1$, $O_{0;\vec{x}}=\prod_{i\in N_{r}\cap N_{t_{2}}}A_{i;1}$, $O_{1;\vec{y}}=A_{r;a}\prod_{i\in N_{r}\setminus N_{t_{2}}}$, and $O_{2;\vec{z}}=A_{t_{2};0}\prod_{i\in N_{t_{2}}\setminus N_{r}}A_{i;1}$, which yields
\begin{equation}
\left\langle A_{r;a}\prod_{i\in N_{r}}A_{i;1} \right\rangle+\left\langle A_{t_{2};0}\prod_{j\in N_{t_{2}}}A_{j;1}\right\rangle\leqslant \left\langle A_{r;a}A_{t_{2};0}\prod_{i\in N_{r}\Delta N_{t_{2}}}A_{i;1} \right\rangle+1.
\end{equation}
We then continue to apply Lemma \ref{lem: general_ineq_expected} recursively with $\alpha_{0},\alpha_{1}=1$. Let us denote by $\langle O_{i}O_{j}\rangle^{(u)}$ the expected value $\langle O_{i}O_{j}\rangle$ in the $u$'th step in this recursion. Then we take $\langle O_{0;\vec{x}}O_{2;\vec{z}}\rangle^{(u)}=\left\langle A_{t_{2u};0}\prod_{j\in N_{t_{2u}}}A_{j;1}\right\rangle$, and $\langle O_{0;\vec{x}}O_{1;\vec{y}}\rangle^{(u)}=\langle O_{1;\vec{y}}O_{2;\vec{z}}\rangle^{(u-1)}$ with $\langle O_{1;\vec{y}}O_{2;\vec{z}}\rangle^{(0)}=\left\langle A_{r;a}\prod_{i\in N_{r}}A_{i;1} \right\rangle$. 

Let us note here that, without specifying observables $O_{i}$ directly, we are opening ourselves to a potential issue: namely that there could be no way of defining observables $O_{i}$ in such a way that we can split the parties into three sets. This could happen if $O_{0;\vec{x}}O_{1;\vec{y}}$ contains $A_{i;k}$ and $O_{0;\vec{x}}O_{2;\vec{z}}$ contains $A_{i;1-k}$ for some $k\in\{0,1\}$, since then $A_{i;k}$ and $A_{i;1-k}$ would need to both belong to $O_{0;\vec{x}}$ and an observable containing two observables for different measurement choices is not well defined. This issue is avoided, however, as we define $t_{0},\ldots, t_{m}$ as the shortest path between $t_{0}$ and $t_{m}$, which directly implies that $t_{i}\notin N_{t_{j}}$ for $|i-j|\geqslant2$, because otherwise a shorter path would exist. This consequently means that $A_{t_{2u};1}$ does not appear as part of any of the expectations $\langle O_{i}O_{j}\rangle$. Moreover, since there are no observables $A_{i;0}$ in Eq. \eqref{eq:first_term} other than $A_{t_{2u};0}$, we can be sure that all operators $O_{0;\vec{x}},O_{1;\vec{y}},O_{2;\vec{z}}$ are well defined.

With this issue resolved we are free to apply Lemma \ref{lem: general_ineq_expected} recursively for all $u\in\{1,\ldots,\lfloor m/2\rfloor\}$ which gives us
\begin{equation}\label{eq:term_1_bound}
\sum_{k=0}^{1}\left\langle A_{r;k}\prod_{i\in N_{r}}A_{i;1} \right\rangle + 2\sum_{u=1}^{\lfloor m/2\rfloor}\left\langle A_{t_{2u};0}\prod_{j\in N_{t_{2u}}}A_{j;1}\right\rangle\leqslant \sum_{k=0}^{1}\left\langle A_{r;k}\prod_{u=1}^{\lfloor m/2\rfloor}A_{t_{2u};0}\prod_{i\in N_{r}\Delta N_{t_{2}}\Delta\ldots\Delta N_{t_{2\lfloor m/2\rfloor}}}A_{i;1} \right\rangle +2\left\lfloor\frac{m}{2}\right\rfloor.
\end{equation}
Note that the last product goes over the set $N_{r}\Delta N_{t_{2}}\Delta\ldots\Delta N_{t_{2\lfloor m/2\rfloor}}$ arises, due to the fact that $A_{i;j}^{2}=\mathbf{1}$, as this set contains parties appearing an odd number of times in the listed neighbourhoods.

We now turn to the second term in Eq. \eqref{eq:three_parts}. Here, as with the first term, we recursively use Lemma \ref{lem: general_ineq_expected}. We take $\alpha_{0},\alpha_{1}=1$, $\langle O_{0;\vec{x}}O_{2;\vec{z}}\rangle^{(u)}=\langle A_{t_{2u+1};0}\prod_{j\in N_{t_{2u+1}}}A_{j;1}\rangle$, and $\langle O_{0;\vec{x}}O_{1;\vec{y}}\rangle^{(n)}=\langle O_{1;\vec{y}}O_{2;\vec{z}}\rangle^{(n-1)}$ with $\langle O_{1;\vec{y}}O_{2;\vec{z}}\rangle^{(0)}=\langle (-1)^{k}A_{r;k}A_{t_{1};0}\prod_{j\in N_{t_{1}}\setminus \{1\}}A_{j;1}\rangle$. As before, due to $t_{0},\ldots,t_{m}$ being the shortest path between $t_{0}$ and $t_{m}$, we have that $t_{i}\notin N_{t_{j}}$ for $|i-j|\geqslant 2$, which in turn implies that $O_{0;\vec{x}},O_{1;\vec{y}},O_{2;\vec{z}}$ are well defined. This all gives us the following upper bound
\begin{equation}\label{eq:term_2_bound}
\begin{aligned}
\langle \tilde{g}_{t_{1}}\rangle+2\sum_{i\in T_{1}}\langle \tilde{g}_{i}\rangle
&=\sum_{k=0}^{1}(-1)^{k}\left\langle A_{r;k}A_{t_{1};0}\prod_{j\in N_{t_{1}}\setminus \{1\}}A_{j;1} \right\rangle + 2\sum_{u=1}^{\lceil m/2\rceil-1}\left\langle A_{t_{2u+1};0}\prod_{j\in N_{t_{2u+1}}}A_{j;1}\right\rangle\\
&\leqslant \sum_{k=0}^{1}(-1)^{k}\left\langle A_{r;k}\prod_{u=0}^{\lceil m/2\rceil -1}A_{t_{2u+1};0}\prod_{i\in (N_{t_{1}}\setminus\{r\})\Delta N_{t_{3}}\Delta\ldots\Delta N_{t_{2\lceil m/2\rceil-1}}}A_{i;1} \right\rangle +2\left(\left\lceil \frac{m}{2}\right\rceil-1\right).
\end{aligned}
\end{equation}

Lastly, we come to the third term in Eq. \eqref{eq:three_parts}. In contrast to the two other parts, here we upper-bound each expected value directly by its absolute maximal value. For $i\in N_{r}\setminus \{t_{1}\}$ each expected value is upper-bounded by $2$ because of the presence of the term $(A_{r;0} - A_{r;1})$, while for each $i\notin \dot{N}_{r}\cup T_{0}\cup T_{1}$ the upper-bound is $1$, but each expected value is multiplied by $2$. Putting this all together, this gives us
\begin{equation}\label{eq:term_3_bound}
\sum_{i\in N_{r}\setminus\{t_{1}\}}\langle \tilde{g}_{i}\rangle+2\sum_{i\notin \dot{N}_{r}\cup T_{0}\cup T_{1}}\langle \tilde{g}_{i}\rangle\leqslant 2(|N_{r}|-1)+2(n-(|N_{r}|+1+|T_{0}|+|T_{1}|))=2(n-1-m),
\end{equation}
where we use the fact that $|T_{0}|+|T_{1}|=m-1$ and that $\dot{N}_{r},T_{0},T_{1}$ are disjoint sets. Putting Eq. \eqref{eq:term_1_bound}, \eqref{eq:term_2_bound}, \eqref{eq:term_3_bound} together we get the following total upper-bound on $\langle I\rangle$
\begin{equation}
\begin{aligned}
\langle I \rangle \leqslant& \left\langle (A_{r;0}+A_{r;1})\prod_{u=1}^{\lfloor m/2\rfloor}A_{t_{2u};0}\prod_{i\in N_{r}\Delta N_{t_{2}}\Delta\ldots\Delta N_{t_{\lfloor m/2\rfloor}}}A_{i;1}\right\rangle
\\
&+ \left\langle (A_{r;0}-A_{r;1})\prod_{u=0}^{\lceil m/2\rceil-1}A_{t_{2u+1};0}\prod_{i\in (N_{t_{1}}\setminus\{r\})\Delta N_{t_{3}}\Delta\ldots\Delta N_{t_{\lceil m\rceil-1}}}A_{i;1} \right\rangle\\
&+ 2(n-2),
\end{aligned}
\end{equation}
where we used the fact that $2\lfloor m/2\rfloor + 2( \lceil m/2\rceil - 1) + 2(n-1-m)=2(n-2)$.
Notice, that our assumption about $t_{0},\ldots,t_{m}$ being the shortest path between $r$ and $l$ implies that $l\in (N_{t_{1}}\setminus\{r\})\Delta N_{t_{3}}\Delta\ldots\Delta N_{t_{m-1+\mu}}\eqqcolon S_{1}$ and that $l\notin N_{r}\Delta N_{t_{2}}\Delta\ldots\Delta N_{t_{m-\mu}}\eqqcolon S_{0}$.  Then, since $r\in Q$ and $l\in \overline{Q}$ then for any behaviour satisfying Eq. \eqref{eq:product_behaviour} we have
\begin{equation}\label{eq:I_second_case}
\begin{aligned}
\langle I \rangle \leqslant& \left\langle (A_{r;0}+A_{r;1})\prod_{i\in T_{0}\cap Q}A_{i;0}\prod_{j\in S_{0}\cap Q}A_{j;1}\right\rangle_{Q} \left\langle\prod_{i\in T_{0}\cap \overline{Q}}A_{i;0}\prod_{j\in S_{0}\cap \overline{Q}}A_{j;1}\right\rangle_{\overline{Q}}\\
&+\left\langle (A_{r;0}-A_{r;1})\prod_{i\in T_{1}\cap Q}A_{i;0}\prod_{j\in S_{1}\cap Q}A_{j;1}\right\rangle_{Q} \left\langle\prod_{i\in T_{1}\cap \overline{Q}}A_{i;0}\prod_{j\in S_{1}\cap \overline{Q}}A_{j;1}\right\rangle_{\overline{Q}}+2(n-2).
\end{aligned}
\end{equation}
We now again use Corollary \ref{cor:bound_2} to upper-bound the above expression in a similar way that it was done in the first case. To this end, we take
\begin{equation}
\begin{aligned}
D_{0;0}=T_{0}\cap T_{1}\cap Q,\quad D_{0;1}=S_{0}\cap S_{1}\cap Q, \quad D_{1;0}= (T_{0}\setminus T_{1})\cap Q,\\
D_{1;1}=(S_{0}\setminus S_{1})\cap Q,\quad D_{2;0}=(T_{1}\setminus T_{0})\cap Q,\quad D_{2;1}=(S_{1}\setminus S_{0})\cap Q.
\end{aligned}
\end{equation}
It is clear that none of the sets $D_{i;j}$ contain $r$. What is more, one can easily see that $D_{i;0}$ are disjoint for all $i$, and that $D_{i;1}$ are disjoint for all $i$. To show that all of $D_{i;j}$ are disjoint, we will show that sets $S_{0}$ and $S_{1}$ do not contain any parties $t_{u}$ except $t_{m}$. This follows from the shortest path assumption, as a party $t_{u}$ is a member of exactly two neighbourhoods of other $t_{i}'s: $$N_{t_{u-1}}$ and $N_{t_{u+1}}$ for all $u\in\{1,\ldots,m-1\}$. Since sets $S_{0}$ and $S_{1}$ are defined as a symmetric difference of neighbourhoods $N_{t_{i}}$ for even or odd $i$'s respectively, none of the parties $t_{u}$ for all $u\in\{1,\ldots,m-1\}$ appear in $S_{0}$ and $S_{1}$. Therefore, the only party from the path $t_{0},\ldots,t_{m}$ that is a part of $S_{0}$ or $S_{1}$ is $t_{m}=l$. 
However, every single $D_{i;j}$ is defined as an intersection with $Q$
but by our assumption $l\in\overline{Q}$. This implies that $D_{i;1}$ do not contain parties $t_{1},\ldots, t_{m}$, and since $D_{i;0}$ are composed solely of $t_{2},\ldots, t_{m-1}$ this shows that all $D_{i;j}$ are disjoint.

Then, we can use Corollary \ref{cor:bound_2}, which yields
\begin{equation}
\begin{aligned}
\langle I \rangle \leqslant& \left\langle (A_{r;0}+A_{r;1})\prod_{i\in T_{0}\cap Q}A_{i;0}\prod_{j\in S_{0}\cap Q}A_{j;1}\right\rangle_{Q} \left\langle\prod_{i\in T_{0}\cap \overline{Q}}A_{i;0}\prod_{j\in S_{0}\cap \overline{Q}}A_{j;1}\right\rangle_{\overline{Q}}\\
&+\left\langle (A_{r;0}-A_{r;1})\prod_{i\in T_{1}\cap Q}A_{i;0}\prod_{j\in S_{1}\cap Q}A_{j;1}\right\rangle_{Q} \left\langle\prod_{i\in T_{1}\cap \overline{Q}}A_{i;0}\prod_{j\in S_{1}\cap \overline{Q}}A_{j;1}\right\rangle_{\overline{Q}}+2(n-2)\\
\leqslant& 2(n-1).
\end{aligned}
\end{equation}

We can now bring both cases $\dot{N}_{r}\cap Q \neq \dot{N}_{r}$ and $\dot{N}_{r}\cap Q = \dot{N}_{r}$ toghether. For both, we showed that for all behaviours satisfying Eq. \eqref{eq:product_behaviour} we get $\langle I\rangle \leqslant2(n-1)$. And so, we can conclude that $\beta_{BL}\leqslant 2(n-1)$. Lastly, to show that this bound is saturable, i.e., $\beta_{BL}= 2(n-1)$, it is sufficient to consider a fully local behaviour for which $\langle A_{i;j}\rangle =1-2\delta_{i,1}\delta_{j,1}$, where $\delta$ is the Kronecker delta, which gives $\langle I\rangle = 2(n-1)$. 

We now proceed with the proof of the quantum bound $\beta_{Q}$. To this end, we consider the following sum-of-squares decomposition
\begin{equation}\label{eq: sos_standard}
\begin{aligned}
(2\sqrt{1+|N_{r}|^{2}} + 2(n-|N_{r}|-1))\mathbb{1} - I =& \frac{1}{\sqrt{1+|N_{r}|^{2}}}\left(\mathbb{1}-\frac{\sqrt{1+|N_{r}|^{2}}}{2}(A_{r;0}+A_{r;1})\prod_{i\in N_{r}}A_{i;1}\right)^{2}\\
&+\sum_{i\in N_{r}}\frac{|N_{r}|}{\sqrt{1+|N_{r}|^{2}}}\left(\mathbb{1}-\frac{\sqrt{1+|N_{r}|^{2}}}{2|N_{r}|}(A_{r;0}-A_{r;1})A_{i;0}\prod_{j\in N_{i}\setminus \{r\}}A_{j;1} \right)^{2}\\
&+\sum_{i\notin \dot{N}_{r}}\left(\mathbb{1}-A_{i;0}\prod_{j\in N_{i}}A_{j;1}\right)^{2},
\end{aligned}
\end{equation}
where $I$ is the Bell operator from Eq. \eqref{eq: gmnl_ineq_general_app}. The proof of the above is straightforward - let us only remark here that the prefactors have been chosen is such a way as to cancel out the anticomutators $\{A_{r;0},A_{r;1}\}$ that originate from the terms $(A_{r;0}+A_{r;1})^{2}$.

The sum-of-squares only implies that $\beta_{Q}\leqslant2\sqrt{1+|N_{1}|^{2}} + 2(n-|N_{1}|-1)$. To show that $\beta_{Q}=2\sqrt{1+|N_{1}|^{2}} + 2(n-|N_{1}|-1)$, we take a graph state $\ket{G}$ and observables
\begin{equation}
A_{r;k}=\frac{X+(-1)^{k}|N_{r}|Z}{\sqrt{1+|N_{r}|^{2}}},\qquad A_{j;0}=X,\qquad A_{j;1}=Z,
\end{equation}
for all $j\in[n]\setminus \{r\}$. Substitution of these observables into $\tilde{g}_{i}$ from Eq. \eqref{eq: gmnl_ineq_general_app} gives us
\begin{equation}
\begin{aligned}
\tilde{g}_{r}&=\left(\frac{X_{r}+|N_{r}|Z_{r}}{\sqrt{1+|N_{r}|^{2}}}+\frac{X_{r}-|N_{r}|Z_{r}}{\sqrt{1+|N_{r}|^{2}}}\right)\prod_{j\in N_{r}} Z_{j} =\frac{2}{\sqrt{1+|N_{r}|^{2}}} g_{r},\\
\forall_{i\in N_{r}},\quad\tilde{g}_{i}&=X_{i}\left(\frac{X_{r}+|N_{r}|Z_{r}}{\sqrt{1+|N_{r}|^{2}}}-\frac{X_{r}-|N_{r}|Z_{r}}{\sqrt{1+|N_{r}|^{2}}}\right)\prod_{j\in N_{i}\setminus\{r\}} Z_{j}=\frac{2|N_{r}|}{\sqrt{1+|N_{r}|^{2}}} g_{i},\\
\forall_{i\notin \dot{N}_{r}},\quad \tilde{g}_{i}&=X_{i}\prod_{j\in N_{i}}Z_{j}=g_{i}.
\end{aligned}
\end{equation}
Hence $\langle I\rangle$ is evaluated to be
\begin{equation}\label{eq:sos_thm_1}
\begin{aligned}
\langle I\rangle =& \sum_{i\in \dot{N}_{r}}\langle \tilde{g}_{i}\rangle+2\sum_{i\notin \dot{N}_{r}}\langle \tilde{g}_{i}\rangle\\
&= \frac{2}{\sqrt{1+|N_{r}|^{2}}} \bra{G}g_{r}\ket{G}+ \sum_{i\in N_{r}}\frac{2|N_{r}|}{\sqrt{1+|N_{r}|^{2}}} \bra{G}g_{i}\ket{G}+\sum_{i\notin \dot{N}_{r}}\bra{G}g_{i}\ket{G}\\
&=\frac{2}{\sqrt{1+|N_{r}|^{2}}}+\frac{2|N_{r}|^{2}}{\sqrt{1+|N_{r}|^{2}}}+2(n-|N_{r}|-1)=2\sqrt{1+|N_{r}|^{2}}+2(n-|N_{r}|-1),
\end{aligned}
\end{equation}
and so $\beta_{Q}=2\sqrt{1+|N_{r}|^{2}}+2(n-|N_{r}|-1)$ ending the proof.
\end{proof}

\section{Self-testing proof}\label{app: s-t}
Here we show that a behaviour that archives the maximal quantum value of the Bell functional from Theorem \ref{thm: gmnl_standard} can be used to self-test the state and measurements. First, let us recall the following lemma from Ref. \cite{PhysRevLett.124.020402} that is vital for the following proof of self-testing.
\begin{lem}\label{lem: s-t}
Let us consider two, finitely dimensional, unitary operators $M_{1},M_{2}\in\mathcal{B}(\mathcal{H})$. If $M_{1}^{2}=M_{2}^{2}=\mathbb{1}$ and they anticommute $\{M_{1},M_{2}\}=0$, then there exists a unitary $U\in\mathcal{B}(\mathcal{H})$ such that 
\begin{equation}
U M_{1}U^{\dagger}=X\otimes \mathbb{1},\quad UM_{2}U^{\dagger}=Z\otimes \mathbb{1}, 
\end{equation}
where $X\in\mathbb{C}^{2}$, $\mathbb{1}\in \mathcal{H}'$, where $\mathcal{H}=\mathbb{C}^{2}\otimes\mathcal{H}'$.
\end{lem}

\begin{thm}
Let $\ket{\psi_{\max}}\in \mathcal{H}=\bigotimes_{i=1}^{n}\mathcal{H}_{i}$ be an $n$-partite state and $A_{i;j}$ be finitely dimensional observables measured by party $i$ that give rise to the maximal quantum violation of the Bell functional \eqref{eq: gmnl_ineq_general} for a given connected graph $G$. Moreover, let us assume that a reduced state $\operatorname{Tr}_{[n]\setminus\{i\}}[\ket{\psi_{\max}}\!\bra{\psi_{\max}}]$ for all parties $i\in[n]$ is full-rank. Then, there exists a unitary operation $U=U_1\otimes\ldots\otimes U_n$ with $U_i$ being a local unitary such that $U_i:\mathcal{H}_i\to \mathbbm{C}^2\otimes\mathcal{H}_i'$ for some auxiliary Hilbert space $\mathcal{H}_i'$ and
\begin{equation}
\begin{aligned}
U\ket{\psi_{\max}}&=\ket{G}\ket{\mathrm{aux}},\\ U_i A_{i;j}\, U_i^{\dagger}&=X^{1-j}_{i}Z_{i}^{j}\otimes \mathbb{1}_{i}\quad \mathrm{for }\; i\in[n]\setminus\{r\},  \\
U_r A_{r;j}\,U_r^{\dagger}&=\frac{X_{r}+(-1)^{j}|N_{r}|Z_{r}}{\sqrt{1+|N_{r}|^{2}}}\otimes \mathbb{1}_{r},
\end{aligned}
\end{equation}
where $\ket{G}\in (\mathbb{C}^{2})^{\otimes n}$ is the graph state to be self-tested, $\ket{\mathrm{aux}}\in \bigotimes_{i=1}^{n}\mathcal{H}_{i}'$, and $\mathbb{1}_{i}$ is an identity operator acting on $\mathcal{H}_{i}'$.
\end{thm}
\begin{proof}
From the sum-of-square decomposition \eqref{eq: sos_standard} it follows that for the state and measurement leading to the maximal quantum value $\beta_{Q}$ we have
\begin{equation}
\begin{aligned}
\bra{\psi_{\max}}\left(\mathbb{1}-\frac{\sqrt{1+|N_{r}|^{2}}}{2}(A_{r;0}+A_{r;1})\prod_{i\in N_{r}}A_{i;1}\right)^{2}\ket{\psi_{\max}}=&0,\\
\bra{\psi_{\max}}\left(\mathbb{1}-\frac{\sqrt{1+|N_{r}|^{2}}}{2|N_{r}|}(A_{r;0}-A_{r;1})A_{i;0}\prod_{j\in N_{i}\setminus \{r\}}A_{j;1} \right)^{2}\ket{\psi_{\max}}=&0,\quad\textrm{for all }i\in N_{r},\\
\bra{\psi_{\max}}\left(\mathbb{1}-A_{i;0}\prod_{j\in N_{i}}A_{j;1}\right)^{2}\ket{\psi_{\max}}=&0,\quad \textrm{for all }i\notin \dot{N}_{r},
\end{aligned}
\end{equation}
which implies
\begin{equation}\label{eq: stabilizing_observables}
\begin{aligned}
\frac{\sqrt{1+|N_{r}|^{2}}}{2}(A_{r;0}+A_{r;1})\prod_{i\in N_{r}}A_{i;1}\ket{\psi_{\max}}=&\ket{\psi_{\max}},\\
\frac{\sqrt{1+|N_{r}|^{2}}}{2|N_{r}|}(A_{r;0}-A_{r;1})A_{i;0}\prod_{j\in N_{i}\setminus \{r\}}A_{j;1} \ket{\psi_{\max}}=&\ket{\psi_{\max}},\quad\textrm{for all }i\in N_{r},\\
A_{i;0}\prod_{j\in N_{i}}A_{j;1}\ket{\psi_{\max}}=&\ket{\psi_{\max}},\quad \textrm{for all }i\notin \dot{N}_{r}.
\end{aligned}
\end{equation}
Our goal is to utilize lemma to show that the above relations, are equivalent up to some unitary $U$ to stabilizing relations of a graph state $\ket{G}$. To this end, we need to show that $\tilde{X}^{2}_{i}=\tilde{Z}^{2}_{i}=\mathbb{1}$ and $\{\tilde{X_{i}},\tilde{Z}_{i}\}=0$, where $\tilde{X}_{i}=A_{i;0}$, $\tilde{Z}_{i}=A_{i;1}$ for $i\in [n]\setminus \{r\}$, and 
\begin{equation}
   \tilde{X}_{r}= \frac{\sqrt{1+|N_{r}|^{2}}}{2}(A_{r;0}+A_{r;1}),\qquad \tilde{Z}_{r}=\frac{\sqrt{1+|N_{r}|^{2}}}{2|N_{r}|}(A_{r;0}-A_{r;1}).
\end{equation}

Let us begin from the proof of $\tilde{X}^{2}_{i}=\tilde{Z}^{2}_{i}=\mathbb{1}$. From our construction, each observable is dichotomic. What is more, we do not constrain the dimension of the Hilbert space $\mathcal{H}$, which allows us to assume that the measurements are projective, implying that the observables are unitary. Due to this, $\tilde{X}^{2}_{i}=\tilde{Z}^{2}_{i}=\mathbb{1}$ is satisfied by construction for all $i\neq r$. As for the party $r$, from Eq. \eqref{eq: stabilizing_observables} it follows that
\begin{equation}
\begin{aligned}
\tilde{X}_{r}\prod_{i\in N_{r}}\tilde{Z}_{r}\cdot|\qquad\tilde{X}_{r}\prod_{i\in N_{r}}\tilde{Z}_{r}\ket{\psi_{\max}}&=\ket{\psi_{\max}},\\
\tilde{X}_{r}^{2}\prod_{i\in N_{r}}\tilde{Z}_{r}^{2}\ket{\psi_{\max}}&=\tilde{X}_{r}\prod_{i\in N_{r}}\tilde{Z}_{r}\ket{\psi_{\max}},\\
\tilde{X}_{r}^{2}\ket{\psi_{\max}}&=\ket{\psi_{\max}}.\\
\end{aligned}
\end{equation}
However, since we assume that the reduced state of $\ket{\psi_{\max}}$ is full-rank, the above implies that $\tilde{X}_{r}^{2}=\mathbb{1}$. One can prove $\tilde{Z}_{r}^{2}=\mathbb{1}$ in the same way by considering a stabilising relation where a neighbour of $r$ acts with $\tilde{X}_{i}$.

Next, let us show that $\{\tilde{X_{i}},\tilde{Z}_{i}\}=0$. First, from construction, we have that $\{\tilde{X}_{r},\tilde{Z}_{r}\}=0$. For any neighbour $i\in N_{r}$ we can use this fact as follows
\begin{equation}
\begin{aligned}
\{\tilde{X}_{i},\tilde{Z}_{i}\}\ket{\psi_{\max}}&=(\tilde{X}_{r}\prod_{j\in N_{r}\setminus \{i\}}\tilde{Z}_{j}\cdot\tilde{Z}_{r}\prod_{j\in N_{i}\setminus\{r\}}\tilde{Z}_{j}+\tilde{Z}_{r}\prod_{j\in N_{i}\setminus\{r\}}\tilde{Z}_{j}\cdot\tilde{X}_{r}\prod_{j\in N_{r}\setminus \{i\}}\tilde{Z}_{j})\ket{\psi_{\max}}\\
&=\prod_{j\in (N_{r}\triangle N_{j})\setminus\{r,j\}} \tilde{Z}_{j} \cdot \{\tilde{X}_{r},\tilde{Z}_{r}\}\ket{\psi_{\max}}=0,
\end{aligned}
\end{equation}
where $N_{r}\triangle N_{j}=N_{r}\cup N_{j}\setminus N_{r}\cap N_{j}$. Due to our assumption, that the reduced state $\operatorname{Tr}_{[n]\setminus\{i\}}[\ket{\psi_{\max}}\!\bra{\psi_{\max}}]$ is full-rank for all parties $i\in[n]$, we have that $\{\tilde{X}_{i},\tilde{Z}_{i}\}=0$. Importantly, the same idea can now be repeated again for the neighbours of neighbours of $r$, and so on. And since we assume that $G$ is connected, we can use this procedure to show $\{\tilde{X}_{i},\tilde{Z}_{i}\}=0$ for all $i\in [n]$.

Since all of the conditions of Lemma \ref{lem: s-t} are satisfied, it follows that there exist a unitary $U=\bigotimes_{i=1}^{n}U_{i}$ such that
\begin{equation}
U \tilde{X}_{i} U^{\dagger} = X_{i}\otimes \mathbb{1}_{i},\quad U \tilde{Z}_{i} U^{\dagger} = Z_{i}\otimes \mathbb{1}_{i},
\end{equation}
for all $i\in [n]$. Applying this to Eq. \eqref{eq: stabilizing_observables} yields
\begin{equation}
 X_{i}\prod_{j\in N_{i}}Z_{j}\bigotimes_{k\in N_{i}\cup\{i\}}\mathbb{1}_{k}\ket{\phi}=\ket{\phi},
\end{equation}
where $\ket{\phi}=U\ket{\psi_{\max}}$. The above relations uniquely identify a state acting on $(\mathbb{C}^{2})^{\otimes n}$ as the graph state $\ket{G}$, which means that $U\ket{\psi_{\max}}=\ket{G}\ket{\textrm{aux}}$ for some $\ket{\textrm{aux}}\in \bigotimes_{i=1}^{n}\mathcal{H}_{i}'$, which ends the proof.

\end{proof}

\section{Proof of Theorem \ref{thm: clique}}\label{app: clique}

In this section, we prove Theorem \ref{thm: clique}. To this end, we will repeatedly utilise Lemma \ref{lem: general_ineq_expected}, which offers a challenge. The upper bound this lemma provides also gives a construction of the upper-bounding observable. However, the construction relies on identifying observables $O_{0,\vec{x}},O_{1;\vec{y}},O_{2;\vec{z}}$ for each application of Lemma \ref{lem: general_ineq_expected}, which with the large number of applications becomes cumbersome. An alternative way of identifying the upper-bounding observable is to construct it directly via an application of the substitution \eqref{eq: substitution_graph_state_ineq} to a product of stabilizing operators. This is the subject of the following lemma.
\begin{lem}\label{lem:observable_from_pauli}
Let $h_{0}$ and $h_{1}$ be Pauli strings containing at most one of $X$ and $Z$, at each site. Moreover, let $\tilde{h}_{0}$, $\tilde{h}_{1}$ be observables constructed from $h_{0}$ and $h_{1}$ via the substitution \eqref{eq: substitution_graph_state_ineq}, and let us assume that we can apply Lemma \ref{lem: general_ineq_expected} to the sum of $\langle \tilde{h}_{0}\rangle$ and $ \langle \tilde{h}_{1}\rangle$. Then, in the resulting upper bound
\begin{equation}
\gamma_{1}\langle \tilde{h}_{0}\rangle + \gamma_{0}\langle \tilde{h}_{1}\rangle\leqslant\langle \tilde{h}_{2}\rangle +\gamma_{0}\gamma_{1},
\end{equation}
the observable $\tilde{h_{2}}$ is constructed via the substitution \eqref{eq: substitution_graph_state_ineq} from the Pauli string $h_{2}=h_{0}h_{1}$, where $\gamma_{0},\gamma_{1}\in {1,2}$ are coefficients chosen based on the presence of the rotated party $r$ in $\tilde{h}_{0}$ and $\tilde{h}_{1}$ - if $\tilde{h}_{i}$ acts nontrivially on $r$ then $\gamma_{i}=2$, otherwise $\gamma_{i}=1$.
\end{lem}
\begin{proof}
Let $h_{0}=\prod_{i\in R_{0;0}}X_{i}\prod_{j\in R_{0;1}}Z_{j}$, and $h_{1}=\prod_{i\in R_{1;0}}X_{i}\prod_{j\in R_{1;1}}Z_{j}$ for some sets $R_{i,j}$ with $i,j\in\{0,1\}$. Since we assume that 
$h_{0}$ and $h_{1}$ do not contain a matrix product of $X$ and $Z$ we have that $R_{i;0}\cap R_{i;1}=\emptyset$ for all $i\in\{0,1\}$. Moreover, let us first consider the case where $\tilde{h}_{0}$ and $\tilde{h}_{1}$ do not contain $(A_{r;0}\pm A_{r;1})$. Then the observables are given as
\begin{equation}
\tilde{h}_{0}=\prod_{i\in R_{0;0}}A_{i;0}\prod_{j\in R_{0;1}}A_{j;1},\quad \tilde{h}_{1}=\prod_{i\in R_{1;0}}A_{i;0}\prod_{j\in R_{1;1}}A_{j;1}.
\end{equation}
For Lemma \ref{lem: general_ineq_expected} to be applicable to $\langle \tilde{h}_{0}\rangle + \langle \tilde{h}_{1}\rangle$, we need that $R_{i;0}\cap R_{1-i;1}=\emptyset$ for all $i\in\{0,1\}$, since otherwise there would not exist a splitting of parties into three disjoint sets, or the observables would not be well-defined. To use Lemma \ref{lem: general_ineq_expected} we take the three observables as
\begin{equation}\label{eq:three_O_def}
O_{0,\vec{x}}=\prod_{i\in R_{0;0}\cap R_{1;0}}A_{i;0}\prod_{j\in R_{0;1}\cap R_{1;1}}A_{j;1},\; O_{1;\vec{y}}=\prod_{i\in R_{0;0}\setminus R_{1;0}}A_{i;0}\prod_{j\in R_{0;1}\setminus R_{1;1}}A_{j;1},\; O_{2;\vec{z}}=\prod_{i\in R_{1;0}\setminus R_{0;0}}A_{i;0}\prod_{j\in R_{1;1}\setminus R_{0;1}}A_{j;1}.
\end{equation}
Since $\tilde{h}_{2}=O_{1,\vec{y}}O_{2,\vec{z}}$, we get
\begin{equation}
\tilde{h}_{2}=\prod_{i\in R_{0;0}\Delta R_{1;0}}A_{i;0}\prod_{j\in R_{0;1}\Delta R_{1;1}} A_{j;1},
\end{equation}
where $\Delta$ is a symmetric difference. On the other hand, the product $h_{0}h_{1}$ equals
\begin{equation}
h_{0}h_{1}=\prod_{i\in R_{0;0}}X_{i}\prod_{j\in R_{0;1}}Z_{j}\prod_{i\in R_{1;0}}X_{i}\prod_{j\in R_{1;1}}Z_{j}=\prod_{i\in R_{0;0}\Delta R_{1;0}}X_{i}\prod_{j\in R_{0;1}\Delta R_{1;1}}Z_{j},
\end{equation}
where we used the fact that $X^{2}=Z^{2}=\mathbb{1}$, and that by assumption $R_{i;0}\cap R_{1-i;1}=\emptyset$ for all $i\in\{0,1\}$. Application of the substitution \eqref{eq: substitution_graph_state_ineq} to the above finishes the first part of the proof. 

Next, let us consider the case where $\tilde{h}_{0}$ contains $(A_{r;0}\pm A_{r;1})$. This requires us to use Lemma \ref{lem: general_ineq_expected} twice: once for each $A_{r;i}$. Let $\tilde{h}_{0,i}$ for $i\in\{0,1\}$ be observables defined as
\begin{equation}
\tilde{h}_{0,0}=A_{r;0}\prod_{i\in R_{0;0}\setminus\{r\}}A_{i;0}\prod_{j\in R_{0;1}\setminus\{r\}}A_{j;1},\quad \tilde{h}_{0,1}=A_{r;1}\prod_{i\in R_{1;0}\setminus\{r\}}A_{i;0}\prod_{j\in R_{1;1}\setminus\{r\}}A_{j;1},
\end{equation}
which gives us $\tilde{h}_{0,0}\pm \tilde{h}_{0,1}=\tilde{h}_{0}$. To use Lemma \ref{lem: general_ineq_expected} we take the three observables as defined in Eq. \eqref{eq:three_O_def}. Then we get
\begin{equation}
\begin{aligned}
\langle \tilde{h}_{0,0}\rangle+\langle\tilde{h}_{1}\rangle \leqslant \langle\tilde{h}_{2,0}\rangle +1,\qquad \pm\langle \tilde{h}_{0,1}\rangle+\langle\tilde{h}_{1}\rangle \leqslant \pm\langle\tilde{h}_{2,1}\rangle +1,
\end{aligned}
\end{equation}
where
\begin{equation}
\tilde{h}_{2,0}=A_{r;0}\prod_{i\in (R_{0;0}\Delta R_{1;0})\setminus\{r\}}A_{i;0}\prod_{j\in (R_{0;1}\Delta R_{1;1})\setminus \{r\}} A_{j;1},\quad \tilde{h}_{2,1}=A_{r;1}\prod_{i\in (R_{0;0}\Delta R_{1;0})\setminus\{r\}}A_{i;0}\prod_{j\in (R_{0;1}\Delta R_{1;1})\setminus \{r\}} A_{j;1}.
\end{equation}
Then we take $\tilde{h}_{2}=\tilde{h}_{2,0}\pm \tilde{h}_{2,1}$, and as in the previous case, we can see that $\tilde{h}_{2}$ is equal to an observable derived from $h_{0}h_{1}$ under the substitution \eqref{eq: substitution_graph_state_ineq}. The case in which $\tilde{h}_{1}$ acts nontrivially on $r$ can be proven in the same way. This ends the proof.
\end{proof}

We can now move to the proof of the Theorem. Let us restate it below for completeness.

\begin{thm}
Let $G$ be a graph such that there exists a vertex $r$ for which the induced subgraph of its neighbourhood $N_{r}$ forms a complete graph. Then, we construct the following functional associated with $G$
\begin{equation}\label{eq: gmnl_ineq_clique_app}
\sum_{\substack{S\subseteq \dot{N}_{r}\\ |S|\textrm{ - odd}}} \langle \tilde{f}_{S}\rangle +2^{|\dot{N}_{r}|-1} \sum_{k\in [n]\setminus \dot{N}_{r}}\langle\tilde{g}_{k}\rangle,
\end{equation}
where $\dot{N}_{r}=N_{r}\cup\{r\}$, and $f_{S}=\prod_{i\in S}g_{i}$ which under Eq. \eqref{eq: substitution_graph_state_ineq} yields
\begin{equation}
\tilde{f}_{S}= (-1)^{\frac{|S|-1}{2}}(A_{r;0}+\xi A_{r;1})\prod_{i\in S\setminus\{r\}}A_{i;0}\prod_{j\in \underset{k\in S}{\Delta}N_{k}\setminus\{r\}}A_{j;1}
\end{equation}
with $\xi=1$ if $r\in S$ and $\xi=-1$ for $r\notin S$. For this functional we have $\beta_{BL}=2^{|\dot{N}_{r}|-1}(n-|\dot{N}_{r}|+1)$ and $\beta_{Q}=2^{|\dot{N}_{r}|-1}(n-|\dot{N}_{r}|+\sqrt{2})$.
\end{thm}

\begin{proof}
Before we proceed with the proof of the bounds for bi-local and quantum behaviours, we first need to show that $\tilde{f}_{S}$ are well-defined observables. This is because, they are constructed via a substitution \eqref{eq: substitution_graph_state_ineq} which works only if the Pauli string is composed solely of $X$ and $Z$ and not their product. Therefore, we need to prove that this is the case for $f_{S}$.

Let $f_{S}^{\{j\}}$ be a truncation of $f_{S}$ to the qubit $j$, i.e., the part of $f_{S}$ that acts on qubit $j$. Then for all $j\notin S$ we have 
\begin{equation}
f_{S}^{\{j\}}=\begin{cases}
Z &\textrm{if}\quad  j\in\underset{i\in S}{\Delta}N_{i}, \\
\mathbb{1} &\textrm{if}\quad  j\notin\underset{i\in S}{\Delta}N_{i} \cup S,
\end{cases}
\end{equation}
where we used the fact that $S\cap \underset{i\in S}{\Delta}N_{i}=\emptyset$ which follows from $|S|$ being odd and from $N_{r}$ being a clique. Therefore, the substitution \eqref{eq: substitution_graph_state_ineq} on $j\notin S$ is well-defined. 

Next, let us consider $f_{S}^{S}$, so the part of $f_{S}$ that is acting on qubits from $S$. Let us denote $S= \{s_{j}\}_{j=1}^{|S|}$. Then, since $|S|$ is odd, we have
\begin{equation}
f_{S}^{S}=\prod_{j=1}^{|S|}Z_{s_{j}}^{j-1}X_{s_{j}}Z^{|S|-j}_{s_{j}}=\prod_{j=1}^{|S|}(-1)^{j-1}X_{s_{j}}=(-1)^{\sum_{k=0}^{|S|-1}k}\prod_{j=1}^{|S|}X_{s_{j}}=(-1)^{\frac{(|S|-1)|S|}{2}}\prod_{j=1}^{|S|}X_{s_{j}}=(-1)^{\frac{|S|-1}{2}}\prod_{j=1}^{|S|}X_{s_{j}},
\end{equation}
and so the substitution \eqref{eq: substitution_graph_state_ineq} is well defined on $j\in S$. Taken all together $\tilde{f}_{S}$ equals
\begin{equation}\label{eq:f_tilde}
\tilde{f}_{S}= (-1)^{\frac{|S|-1}{2}}(A_{r;0}+\xi A_{r;1})\prod_{i\in S\setminus\{r\}}A_{i;0}\prod_{j\in \underset{k\in S}{\Delta}N_{k}\setminus\{r\}}A_{j;1},
\end{equation}
where
\begin{equation}
\xi = \begin{cases}
1 \quad \textrm{if}\quad r\in S,\\
-1\quad \textrm{if}\quad r\notin S.\\
\end{cases}
\end{equation}
In the next step of the proof, we show that the bi-local bound for the functional \eqref{eq: gmnl_ineq_clique_app} equals $\beta_{BL}=2^{|\dot{N}_{r}|-1}+2^{|\dot{N}_{r}|}(n-|\dot{N}_{r}|)$. As in the proof of Theorem \ref{thm: gmnl_standard}, here we also optimise over bipartitions $Q|\overline{Q}$ and behaviours that are product across $Q|\overline{Q}$ (see Eq. \eqref{eq:product_behaviour} for the definition). For such behaviors, we have that the functional $\langle I\rangle$ given by Eq. \eqref{eq: gmnl_ineq_clique_app} is expressed as 
\begin{equation}
\langle I\rangle=\sum_{\substack{S\subseteq \dot{N}_{r}\\ |S|\textrm{ - odd}}} \left\langle \tilde{f}_{S}^{Q}\right\rangle_{Q}\left\langle \tilde{f}_{S}^{ \overline{Q}}\right\rangle_{\overline{Q}} +2^{|\dot{N}_{r}|-1} \sum_{k\in [n]\setminus \dot{N}_{r}}\left\langle\tilde{g}_{k}^{Q}\right\rangle_{Q}\left\langle\tilde{g}_{k}^{\overline{Q}}\right\rangle_{\overline{Q}},
\end{equation}
where $\langle \cdot \rangle_{Q}$ is evaluated over $p(\mathbf{a}_{Q}|\mathbf{x}_{Q})$ and likewise for $\langle \cdot \rangle_{\overline{Q}}$, and the upper index $Q/\overline{Q}$ denotes the truncation of a given operator to observables acting only on the set of parties $Q/\overline{Q}$. Without loss of generality, we assume that $r\in Q$.

We split the rest of the proof into two cases: $\dot{N}_{r}\cap Q \neq \dot{N}_{r}$ and $\dot{N}_{r}\cap Q = \dot{N}_{r}$, which corresponds to the scenario in which $Q|\overline{Q}$ either cuts across the core of the inequality or it does not. We divide it in such a way, as we can only use Corollary \ref{cor:bound_2} to derive a non-trivial bound for a behaviour local across $Q|\overline{Q}$ if the bipartition cuts across the core of the inequality. This is true for the first case, but not for the second. Therefore, in the first case we can apply our proof directly, whereas for the second case we first need to identify a party $l$ that is in a different partition than the core. Then, using Lemma \ref{lem: general_ineq_expected}, we "move" this party to the core by upper bounding the original inequality by another one in which some of the parties in the core are replaced by $l$, enabling us to use the same strategy as for the first part of the proof.

We begin with the former. For this case, we can upper bound all expected values $\langle\tilde{g}_{k}\rangle$ by $1$. This leaves us with the following expression
\begin{equation}\label{eq:sum_S_1}
\langle I\rangle\leqslant\sum_{\substack{S\subseteq \dot{N}_{r}\\ |S|\textrm{ - odd}}} \langle \tilde{f}_{S}^{Q}\rangle_{Q}\langle \tilde{f}_{S}^{\overline{Q}}\rangle_{\overline{Q}} + 2^{|\dot{N}_{r}|-1}(n-|\dot{N}_{r}|).
\end{equation}
To prove the upper bound for this expression, we wish to apply Corollary \ref{cor:bound_2} to pairs of terms in the above sum. To this end, we first write down the contribution of party $r$ explicitly:
\begin{equation}\label{eq:sum_S_2}
\langle I\rangle\leqslant\sum_{\substack{S\subseteq N_{r}\\ |S|\textrm{ - even}}} \langle (A_{r;0}+A_{r;1})\tilde{f}_{S\cup\{r\}}^{Q\setminus\{r\}}\rangle_{Q}\langle \tilde{f}_{S\cup\{r\}}^{\overline{Q}}\rangle_{\overline{Q}} +\sum_{\substack{S\subseteq N_{r}\\ |S|\textrm{ - odd}}} \langle (A_{r;0}-A_{r;1})\tilde{f}_{S}^{Q\setminus\{r\}}\rangle_{Q}\langle \tilde{f}_{S}^{\overline{Q}}\rangle_{\overline{Q}} + 2^{|\dot{N}_{r}|-1}(n-|\dot{N}_{r}|).
\end{equation}
where the first sum originates from the sets $S$ in Eq. \eqref{eq:sum_S_1} that contain $r$ and the second one from the sets that do not contain $r$. Notice that, by Eq. \eqref{eq:f_tilde}, the splitting separates exactly all the terms in which the rotated party contributes with $A_{r;0} + A_{r;1}$" and $A_{r;0} - A_{r;1}$. Notice the change of the condition on $|S|$ in the first sum - this is due to the fact that if the original set of odd cardinality contained $r$, then taking this party out of the set makes its cardinality even.

The next step is to find a pairing between terms from the first and the second sum that allows us to use Corollary \ref{cor:bound_2}. To satisfy the condition of this corollary, we can apply the following strategy: we pick the terms such that the subset of parties that act with $A_{i;0}$ in $\tilde{f}_{S\cup\{r\}}^{Q\setminus\{r\}}$ (from the first sum) and in $\tilde{f}_{S}^{Q\setminus\{r\}}$ (from the second sum) are the same in both cases. This makes it so that sets $D_{1;0}$ and $D_{2;0}$ from Corollary \ref{cor:bound_2} are empty, which implies that it is always possible to split the parties into sets $D_{i;j}$, allowing us to use Corollary \ref{cor:bound_2}. This can be achieved by utilising Eq. \eqref{eq:f_tilde}, from which it follows that parties acting with $A_{i;0}$ in $\tilde{f}_{S\cup\{r\}}^{Q\setminus\{r\}}$ and $\tilde{f}_{S}^{Q\setminus\{r\}}$ belong to the set $S\cap Q$.

To find such a pairing of the terms from both sums in Eq. \eqref{eq:sum_S_2}, let us first note that by our assumption $\dot{N}_{r}\cap Q\neq \dot{N}_{r}$, there has to exist a party $l$ such that $l\in \dot{N}_{r}$ and $l\in\overline{Q}$. Writing down the contribution of party $l$ explicitly yields
\begin{equation}\label{eq:two_terms_ghz}
\begin{aligned}
\langle I\rangle\leqslant&\sum_{\substack{S\subseteq N_{r}\setminus\{l\}\\ |S|\textrm{ - even}}} \left(\langle(A_{r;0}+A_{r;1}) \tilde{f}_{S\cup\{r\}}^{Q\setminus\{r\}}\rangle_{Q}\langle A_{l;1} \tilde{f}_{S\cup\{r\}}^{\overline{Q}\setminus\{l\}}\rangle_{\overline{Q}}+\langle(A_{r;0}-A_{r;1}) \tilde{f}_{S\cup\{l\}}^{Q\setminus\{r\}}\rangle_{Q}\langle A_{l;0}\tilde{f}_{S\cup\{l\}}^{\overline{Q}\setminus\{l\}}\rangle_{\overline{Q}}\right)\\
&+\sum_{\substack{S\subseteq N_{r}\setminus\{l\}\\ |S|\textrm{ - odd}}}\left( \langle (A_{r;0}+A_{r;1})\tilde{f}_{S\cup\{r,l\}}^{Q\setminus\{r\}}\rangle_{Q}\langle A_{l;0}\tilde{f}_{S\cup\{r,l\}}^{\overline{Q}\setminus\{l\}} \rangle_{\overline{Q}}+\langle (A_{r;0}-A_{r;1})\tilde{f}_{S}^{Q\setminus\{r\}}\rangle_{Q}\langle A_{l;1}\tilde{f}_{S}^{\overline{Q}\setminus\{l\}}\rangle_{\overline{Q}}\right)\\
&+ 2^{|\dot{N}_{r}|-1}(n-|\dot{N}_{r}|).    
\end{aligned}
\end{equation}
Note that we have not transformed the functional: the only thing that we did in this step is that we specified directly to which sets parties $r$ and $l$ belong, and we wrote the observables of party $r$ and $l$ explicitly. However, the way that we arranged the terms is significant, as we will apply Corollary \ref{cor:bound_2} separately to the first and second sum.

Let us first consider a term in the first sum for some given $S\subseteq N_{r}\setminus \{l\}$ s.t. $|S|$ is even. We want to apply Corollary \ref{cor:bound_2} to expected values on the subset $Q$, while we take $\alpha_{0}=\langle A_{l;1}\tilde{f}_{S\cup\{r\}}^{\overline{Q}\setminus\{l\}}\rangle_{\overline{Q}}$, and $\alpha_{1}=\langle A_{l;0}\tilde{f}_{S\cup\{l\}}^{\overline{Q}\setminus\{l\}}\rangle_{\overline{Q}}$. To this end, we need to find a division of parties into disjoint sets $D_{i;j}$. Let us construct $D_{0;0}$ and $D_{0;1}$ by taking the intersection of sets of parties acting with $A_{i;0}$ and $A_{i;1}$ respectively in $\tilde{f}_{S\cup\{r\}}^{Q\setminus\{r\}}$ and $\tilde{f}_{S\cup\{l\}}^{Q\setminus\{r\}}$ (where as a reminder $f_{S}=\prod_{i\in S}g_{i}$), and let us construct the rest of $D_{i;j}$ by taking appropriate differences of these sets. Then we have
\begin{equation}
\begin{aligned}
D_{0;0}&= (S\cup\{r\})\cap (S\cup\{l\})\cap(Q\setminus\{r\})=S\cap Q,\\
D_{1;0}&= ((S\cup\{r\})\setminus (S\cup\{l\}))\cap(Q\setminus\{r\})=\emptyset,\\
D_{2;0}&= ((S\cup\{l\})\setminus (S\cup\{r\}))\cap(Q\setminus\{r\})=\emptyset,\\
D_{0;1}&=\left(\left(\underset{i\in S\cup\{r\}}{\Delta} N_{i}\right)\cap\left(\underset{i\in S\cup \{l\}}{\Delta} N_{i}\right)\right)\cap (Q\setminus \{r\}),\\
D_{1;1}&=\left(\left(\underset{i\in S\cup\{r\}}{\Delta} N_{i}\right)\setminus\left(\underset{i\in S\cup \{l\}}{\Delta} N_{i}\right)\right)\cap (Q\setminus \{r\}),\\
D_{2;1}&=\left(\left(\underset{i\in S\cup \{l\}}{\Delta} N_{i}\right)\setminus \left(\underset{i\in S\cup\{r\}}{\Delta} N_{i}\right)\right)\cap (Q\setminus \{r\}).\\
\end{aligned}
\end{equation}
From construction, $D_{i;1}$ are disjoint for all $i\in\{0,1,2,\}$. Next, since $\tilde{f}_{S\cup\{r\}}^{Q\setminus\{r\}}$ and $\tilde{f}_{S\cup\{l\}}^{Q\setminus\{r\}}$ are well-defined observables, we have that $D_{0;0}D_{0;1},D_{1;0},D_{1;1}$ are disjoint, and so are $D_{0;0}D_{0;1},D_{2;0},D_{2;1}$. And since $D_{1;0},D_{2;0}=\emptyset$, we have that all $D_{i;j}$ are disjoint. Therefore, we can apply Corollary \ref{cor:bound_2}
\begin{equation}
\begin{aligned}
&\alpha_{0}\left\langle (A_{r;0}+A_{r;1})\tilde{f}_{S\cup\{r\}}^{Q\setminus\{r\}}\right\rangle_{Q}+\alpha_{1}\left\langle (A_{r;0}-A_{r;1})\tilde{f}_{S\cup\{l\}}^{Q\setminus\{r\}}\right\rangle_{Q}\leqslant 2.
\end{aligned}
\end{equation}

Coming back to Eq. \eqref{eq:two_terms_ghz}, the above expression appears there in the first sum, and since the above derivation works for any $S$, the whole term can be upper bounded as
\begin{equation}\label{eq:first_term_ghz}
\sum_{\substack{S\subseteq N_{r}\setminus\{l\}\\ |S|\textrm{ - even}}} \left(\langle(A_{r;0}+A_{r;1}) \tilde{f}_{S\cup\{r\}}^{Q\setminus\{r\}}\rangle_{Q}\langle A_{l;1} \tilde{f}_{S\cup\{r\}}^{\overline{Q}\setminus\{l\}}\rangle_{\overline{Q}}+\langle(A_{r;0}-A_{r;1}) \tilde{f}_{S\cup\{l\}}^{Q\setminus\{r\}}\rangle_{Q}\langle A_{l;0}\tilde{f}_{S\cup\{l\}}^{\overline{Q}\setminus\{l\}}\rangle_{\overline{Q}}\right)\leqslant 2\cdot 2^{|N_{r}|-2}=2^{|\dot{N}_{r}|-2}.
\end{equation}

We can now move to a term in the second sum in Eq. \eqref{eq:two_terms_ghz} for some given $S\subseteq N_{r}\setminus\{l\}$ s.t. $|S|$ is odd. The strategy here is the same as with the first term: we use Corollary \ref{cor:bound_2} to derive an upper bound, and so we need to divide the set of parties into the disjoint sets $D_{i,j}$. We construct these sets in the same way as for the first term, which gives us
\begin{equation}
\begin{aligned}
D_{0;0}&= (S\cup\{r,l\})\cap S\cap(Q\setminus\{r\})=S\cap Q,\\
D_{1;0}&= ((S\cup\{r,l\})\setminus S)\cap(Q\setminus\{r\})=\emptyset,\\
D_{2;0}&= (S\setminus (S\cup\{r,l\}))\cap(Q\setminus\{r\})=\emptyset,\\
D_{0;1}&=\left(\left(\underset{i\in S\cup\{r,l\}}{\Delta} N_{i}\right)\cap\left(\underset{i\in S }{\Delta} N_{i}\right)\right)\cap (Q\setminus \{r\}),\\
D_{1;1}&=\left(\left(\underset{i\in S\cup\{r,l\}}{\Delta} N_{i}\right)\setminus\left(\underset{i\in S}{\Delta} N_{i}\right)\right)\cap (Q\setminus \{r\}),\\
D_{2;1}&=\left(\left(\underset{i\in S}{\Delta} N_{i}\right)\setminus \left(\underset{i\in S\cup\{r,l\}}{\Delta} N_{i}\right)\right)\cap (Q\setminus \{r\}).\\
\end{aligned}
\end{equation}
By the same arguments as before, these sets are disjoint. Therefore, we can apply Corollary \ref{cor:bound_2}
\begin{equation}
\alpha_{0}\left\langle (A_{r;0}+A_{r;1})\tilde{f}_{S\cup\{r,l\}}^{Q\setminus\{r\}}\right\rangle_{Q}+\alpha_{1}\left\langle (A_{r;0}-A_{r;1})\tilde{f}_{S}^{Q\setminus\{r\}}\right\rangle_{Q}\leqslant 2,
\end{equation}
where we take $\alpha_{0}=\langle A_{l;0}\tilde{f}_{S\cup\{r,l\}}^{\overline{Q}\setminus\{l\}} \rangle_{\overline{Q}}$, $\alpha_{1}=\langle A_{l;1}\tilde{f}_{S}^{\overline{Q}\setminus\{l\}}\rangle_{\overline{Q}}$. Then the sum of the above over all $S\subseteq N_{r}\setminus\{l\}$ s.t. $|S|$ is odd, gives us
\begin{equation}\label{eq:second_term_ghz}
\sum_{\substack{S\subseteq N_{r}\setminus\{l\}\\ |S|\textrm{ - odd}}}\left( \langle (A_{r;0}+A_{r;1})\tilde{f}_{S\cup\{r,l\}}^{Q\setminus\{r\}}\rangle_{Q}\langle A_{l;0}\tilde{f}_{S\cup\{r,l\}}^{\overline{Q}\setminus\{l\}} \rangle_{\overline{Q}}+\langle (A_{r;0}-A_{r;1})\tilde{f}_{S}^{Q\setminus\{r\}}\rangle_{Q}\langle A_{l;1}\tilde{f}_{S}^{\overline{Q}\setminus\{l\}}\rangle_{\overline{Q}}\right)\leqslant 2\cdot 2^{|N_{r}|-2}=2^{|\dot{N}_{r}|-2}.
\end{equation}

Putting Eq. \eqref{eq:first_term_ghz} and Eq. \eqref{eq:second_term_ghz} together yields the following upper bound
\begin{equation}
\begin{aligned}
\langle I\rangle\leqslant& 2^{|\dot{N}_{r}|-2}+2^{|\dot{N}_{r}|-2}+ 2^{|\dot{N}_{r}|-1}(n-|\dot{N}_{r}|)=2^{|\dot{N}_{r}|-1}(n-|\dot{N}_{r}|+1).
\end{aligned}
\end{equation}

Having proven the upper bound for the case $\dot{N}_{r}\cap Q \neq \dot{N}_{r}$, we now consider $\dot{N}_{r}\cap Q = \dot{N}_{r}$. Let $l\in \overline{Q}$. Since the graph $G$ is connected, there exists a path between $r$ and $l$. Let us denote by $t_{0},t_{1},\ldots,t_{m}$ the parties that form the shortest path between $r$ and $l$, where $t_{0}=r$ and $t_{m}=l$. Notice, that without loss of generality, we can assume that $t_{i}\in Q$ for all $i\in \{0,\ldots,m-1\}$, since if $t_{i}\in \overline{Q}$ for some $i\in \{1,\ldots,m-1\}$ then we simply take $l=t_{i}$. In this part, our goal will be to first use the extension part of the inequality and Lemma \ref{lem: general_ineq_expected} to bring $l$ into the core of the inequality. This is because we would like to use an argument from the previous case; however, that only works if the bipartition cuts through the parties in the core of the inequality.  

First, let us consider the part of Eq. \eqref{eq: gmnl_ineq_clique_app} that corresponds to the sum over the parties $t_{u}$ for $u\in \{2,\ldots,m\}$
\begin{equation}
2^{|\dot{N}_{r}|-1}\sum_{u=2}^{m}\langle\tilde{g}_{t_{u}}\rangle=\left(2^{|\dot{N}_{r}|-1}\sum_{i\in T_{0}}\langle\tilde{g}_{i}\rangle\right)+\left(2^{|\dot{N}_{r}|-1}\sum_{i\in T_{1}} \langle\tilde{g}_{i}\rangle\right),
\end{equation}
where $T_{0}=\{t_{2u}\}_{u=1}^{\lfloor m/2\rfloor}$, $T_{1}=\{t_{2u+1}\}_{u=1}^{\lceil m/2\rceil -1}$, and  $\lfloor\cdot\rfloor$, $\lceil \cdot \rceil$ are respectively the floor and ceiling function
Similarly to the proof of Theorem \ref{thm: gmnl_standard}, we tackle the two terms on the right side separately, by recursively applying Lemma \ref{lem: general_ineq_expected}. 

We start from $\sum_{i\in T_{0}} \langle\tilde{g}_{i}\rangle$. Let us conduct the first step explicitly for clarity of the proof, with all other steps relegated to the recursion given afterwards. The first two terms $\sum_{i\in T_{0}}\langle\tilde{g}_{i}\rangle$ are $\langle \tilde{g}_{2}\rangle+\langle\tilde{g}_{4}\rangle$, application of Lemma \ref{lem: general_ineq_expected} to which gives us
\begin{equation}
\langle \tilde{g}_{2}\rangle+\langle\tilde{g}_{4}\rangle\leqslant \langle\tilde{s}\rangle + 1,
\end{equation}
where by Lemma \ref{lem:observable_from_pauli} $s=g_{2}g_{4}$. This will be a pattern repeated for the rest of the proof: we use Lemma \ref{lem: general_ineq_expected} to find an upper bound on a given expression and Lemma \ref{lem:observable_from_pauli} to find a formulation of the upper-bounding observable via the stabilizer formalism.

Next, to apply this strategy recursively we take $\alpha_{0},\alpha_{1}=1$, $\langle O_{0;\vec{x}}O_{2;\vec{z}}\rangle^{(u)}=\langle\tilde{g}_{t_{2u}}\rangle$, $\langle O_{0;\vec{x}}O_{1;\vec{y}}\rangle^{(u)}=\langle O_{1;\vec{y}}O_{2;\vec{z}}\rangle^{(u-1)}$ for $u\in\{2,\ldots,\lfloor m/2\rfloor\}$ with $\langle O_{1;\vec{y}}O_{2;\vec{z}}\rangle^{(1)}=\langle\tilde{g}_{t_{2}}\rangle$ which yields
\begin{equation}
\sum_{i\in T_{0}}\langle\tilde{g}_{i}\rangle\leqslant \langle \tilde{c}_{e}\rangle + \left\lfloor\frac{m}{2} \right\rfloor-1,
\end{equation}
where by Lemma \ref{lem:observable_from_pauli}
\begin{equation}\label{eq:c_e}
c_{e}\coloneqq \prod_{i\in T_{0}}g_{i}.
\end{equation}
Let us remark here that, as in the proof of Theorem \ref{thm: gmnl_standard}, we potentially open ourselves to an issue with the applicability of Lemma \ref{lem: general_ineq_expected} without directly specifying observables $O_{i}$. This could happen if one party acts with different observables in two expected values that we apply Lemma \ref{lem: general_ineq_expected} to. However, this is avoided as the only parties acting with $A_{i;0}$ in the above are $t_{2u}$ for $u\in\{1,\ldots,\lfloor m/2\rfloor\}$, and due to the fact that $t_{0},\ldots,t_{m}$ is the shortest path between $t_{0}$ and $t_{m}$, $t_{2u}\notin N_{t_{2j}}$ for all $u,j\in\{1,\ldots,\lfloor m/2\rfloor\}$, and so none of these parties act with $A_{t_{2u};1}$ in these expected values.

Similarly, we take $\alpha_{0},\alpha_{1}=1$, $\langle O_{0;\vec{x}}O_{2;\vec{z}}\rangle^{(u)}=\langle\tilde{g}_{t_{2u+1}}\rangle$, $\langle O_{0;\vec{x}}O_{1;\vec{y}}\rangle^{(u)}=\langle O_{1;\vec{y}}O_{2;\vec{z}}\rangle^{(u-1)}$ for $u\in\{2,\ldots,\lceil m/2\rceil-1\}$ with $\langle O_{1;\vec{y}}O_{2;\vec{z}}\rangle^{(1)}=\langle\tilde{g}_{t_{3}}\rangle$ we get
\begin{equation}
\sum_{i\in T_{1}} \langle\tilde{g}_{i}\rangle\leqslant \langle \tilde{c}_{o}\rangle + \left\lceil\frac{m}{2} \right\rceil-2,
\end{equation}
where by Lemma \ref{lem:observable_from_pauli}
\begin{equation}\label{eq:c_o}
c_{o}\coloneqq \prod_{i\in T_{1}}g_{i}.
\end{equation}
The applicability of Lemma \ref{lem: general_ineq_expected} is resolved by the same argument based on $t_{0},\ldots,t_{m}$ being the shortest path between $t_{0}$ and $t_{m}$.

Putting these two upper bounds together, we get
\begin{equation}
\sum_{u=2}^{m}\langle \tilde{g}_{t_{u}}\rangle\leqslant \langle \tilde{c}_{e}\rangle+\langle \tilde{c}_{o}\rangle + m-3,
\end{equation}
We apply this bound to the functional $\langle I\rangle$ which gives us
\begin{equation}
\langle I \rangle \leqslant \sum_{\substack{S\subseteq \dot{N}_{r}\\ |S|\textrm{ - odd}}} \langle \tilde{f}_{S}\rangle + 2^{|\dot{N}_{r}|-1}(\langle \tilde{c}_{e}\rangle+\langle\tilde{c}_{o}\rangle+ m-3)+ 2^{|\dot{N}_{r}|-1}\sum_{k\in ([n]\setminus \dot{N}_{r})\setminus(T_{0}\cup T_{1})} \langle\tilde{g}_{k}\rangle .
\end{equation}
We bound each expected value in the last sum by $1$, which gives us
\begin{equation}
\begin{aligned}
\langle I \rangle &\leqslant \sum_{\substack{S\subseteq \dot{N}_{r}\\ |S|\textrm{ - odd}}} \langle \tilde{f}_{S}\rangle + 2^{|\dot{N}_{r}|-1}(\langle \tilde{c}_{e}\rangle+\langle\tilde{c}_{o}\rangle+ m-3)+ 2^{|\dot{N_{r}}|-1} (n-|\dot{N}_{r}|-m+1)\\
&\leqslant \sum_{\substack{S\subseteq \dot{N}_{r}\\ |S|\textrm{ - odd}}} \langle \tilde{f}_{S}\rangle + 2^{|\dot{N}_{r}|-1}\langle \tilde{c}_{e}\rangle+2^{|\dot{N}_{r}|-1}\langle\tilde{c}_{o}\rangle+ 2^{|\dot{N_{r}}|-1} (n-|\dot{N}_{r}|-2).
\end{aligned}
\end{equation}
Notice that now party $l$ has non-trivial contributions in $\langle \tilde{c}_{o}\rangle$ \eqref{eq:c_o} and $\langle \tilde{c}_{e}\rangle$ \eqref{eq:c_e}, whereas we would like it to be a part of the core with terms $\langle \tilde{f}_{S}\rangle$. To achieve this, we need to apply Lemma \ref{lem: general_ineq_expected} two more times. To this end, we rewrite this upper-bound as follows
\begin{equation}\label{eq: split_ineq}
\begin{aligned}
\langle I \rangle \leqslant&   \left(\sum_{\substack{S\subseteq \dot{N}_{r}\\ |S|\textrm{ - odd}\\ t_{1}\notin S}} \langle \tilde{f}_{S}\rangle +2^{|\dot{N}_{r}|-1}\langle\tilde{c}_{e}\rangle \right)+ \left(\sum_{\substack{S\subseteq \dot{N}_{r}\\ |S|\textrm{ - odd}\\ t_{1}\in S}} \langle \tilde{f}_{S}\rangle  + 2^{|\dot{N}_{r}|-1}\langle\tilde{c}_{o}\rangle\right)+ 2^{|\dot{N}_{r}|-1}(n-|\dot{N}_{r}|-2),
\end{aligned}
\end{equation}
where we split the terms $\langle \tilde{f}_{S}\rangle$ into two groups, depending on the presence of $t_{1}$ in $S$. Let us begin with 
\begin{equation}
\sum_{\substack{S\subseteq \dot{N}_{r}\\ |S|\textrm{ - odd}\\ t_{1}\notin S}} \langle \tilde{f}_{S}\rangle +2^{|\dot{N}_{r}|-1}\langle\tilde{c}_{e}\rangle.
\end{equation}
First, notice that from $t_{1}\notin S$ and $\dot{N}_{r}$ being a clique, it follows that each $\tilde{f}_{S}$ contains the observable $A_{t_{1};1}$. What is more, from the fact that $t_{i}$'s form the shortest path between $r$ and $l$ follows that the only terms in common between $\tilde{f}_{S}$ and $\tilde{c}_{e}$ \eqref{eq:c_e} are $A_{i;1}$ for $i\in N_{t_{2}}\cap \underset{j\in S}{\Delta}N_{j}$. This makes it so that all of the terms $\langle \tilde{f}_{S}\rangle$ and $\langle \tilde{c}_{e}\rangle$ satisfy the conditions of Lemma \ref{lem: general_ineq_expected} as we can take $\alpha_{0},\alpha_{1}=1$ and
\begin{equation}
O_{0,\vec{x}}=\prod_{i\in N_{t_{2}}\cap \underset{j\in S}{\Delta}N_{j}} A_{i;1},\quad O_{1;\vec{y}}=(-1)^{\frac{|S|-1}{2}}\prod_{i\in S} A_{i;0}\prod_{j\in( \underset{k\in S}{\Delta}N_{k})\setminus N_{t_{2}}}A_{j;1},\quad O_{2;\vec{z}}=\prod_{i\in T_{0}}A_{i;0}\prod_{j\in( \underset{k\in T_{0}}{\Delta}N_{k})\setminus ( \underset{k\in S}{\Delta}N_{k}) }A_{j;1}.
\end{equation}
Moreover, there are $2^{|\dot{N}_{r}|-2}$ terms $\langle \tilde{f}_{S}\rangle$, each containing two observables acting on $r$, meaning that in total we have to apply Lemma \ref{lem: general_ineq_expected} $2^{|\dot{N}_{r}|-1}$ times
\begin{equation}\label{eq: fs_bound_o}
\sum_{\substack{S\subseteq \dot{N}_{r}\\ |S|\textrm{ - odd}\\t_{1}\notin S}} \langle \tilde{f}_{S}\rangle +2^{|\dot{N}_{r}|-1}\langle\tilde{c}_{e}\rangle\leqslant \sum_{\substack{S\subseteq \dot{N}_{r}\\ |S|\textrm{ - odd}\\ t_{1}\notin S}} \langle \tilde{f}_{e,S}\rangle + 2^{|\dot{N}_{r}|-1},
\end{equation}
where by Lemma \ref{lem:observable_from_pauli}
\begin{equation}\label{eq:f_o}
f_{e,S}=\prod_{i\in S\cup T_{0}}g_{i}.
\end{equation}

Let us now focus on the term
\begin{equation}
\sum_{\substack{S\subseteq \dot{N}_{r}\\ |S|\textrm{ - odd}\\ t_{1}\in S}} \langle \tilde{f}_{S}\rangle  + 2^{|\dot{N}_{r}|-1}\langle\tilde{c}_{o}\rangle.
\end{equation}
To apply Lemma \ref{lem: general_ineq_expected} we take $\alpha_{0},\alpha_{1}=1$ and
\begin{equation}
\begin{aligned}
O_{0;\vec{x}}&=\prod_{j\in R_{1}\cap R_{2}}A_{j;1},\quad
O_{1;\vec{y}}= (-1)^{\frac{|S-1|}{2}}\prod_{i\in S}A_{i;0}\prod_{j\in R_{1}\setminus R_{2}}A_{j;1},\quad 
O_{2;\vec{z}}=\prod_{i\in T_{1}}A_{i;0}\prod_{j\in R_{2}\setminus R_{1}}A_{j;1},
\end{aligned}
\end{equation}
where we take $R_{1}=\underset{k\in S}{\Delta}N_{k}$ and $R_{2}=\underset{k\in T_{1}}{\Delta}N_{k}$. This gives us
\begin{equation}\label{eq: fs_bound_e}
\sum_{\substack{S\subseteq \dot{N}_{r}\\ |S|\textrm{ - odd}\\ t_{1}\in S}} \langle \tilde{f}_{S}\rangle  + 2^{|\dot{N}_{r}|-1}\langle\tilde{c}_{o}\rangle \leqslant \sum_{\substack{S\subseteq \dot{N}_{r}\\ |S|\textrm{ - odd}\\ t_{1}\in S}} \langle \tilde{f}_{o,S}\rangle +2^{|\dot{N}_{r}|-1},
\end{equation}
where by Lemma \ref{lem:observable_from_pauli}
\begin{equation}\label{eq:f_e}
\tilde{f}_{o,S}=\prod_{i\in S \cup  T_{1}}g_{i}.
\end{equation}
Applying Eq. \eqref{eq: fs_bound_o} and \eqref{eq: fs_bound_e} to Eq. \eqref{eq: split_ineq} gives us
\begin{equation}
\begin{aligned}
\langle I \rangle \leqslant&    \sum_{\substack{S\subseteq \dot{N}_{r}\\ |S|\textrm{ - odd}\\ t_{1}\notin S}} \langle \tilde{f}_{e,S}\rangle+ \sum_{\substack{S\subseteq \dot{N}_{r}\\ |S|\textrm{ - odd}\\ t_{1}\in S}} \langle \tilde{f}_{o,S}\rangle+2^{|\dot{N}_{r}|-1}(n-|\dot{N}_{r}|)
.
\end{aligned}
\end{equation}
We can further transform the above expression by excluding $t_{1}$ from $\dot{N}_{r}$ and instead writing it explicitly in the expected values. This gives us 
\begin{equation}
\begin{aligned}
\langle I \rangle \leqslant&   \sum_{\substack{S\subseteq \dot{N}_{r}\setminus\{t_{1}\}\\ |S|\textrm{ - odd}}} \langle \tilde{f}_{e,S}\rangle+ \sum_{\substack{S\subseteq \dot{N}_{r}\setminus\{t_{1}\}\\ |S|\textrm{ - even}}}\langle \tilde{f}_{o,S\cup \{t_{1}\}}\rangle+2^{|\dot{N}_{r}|-1}(n-|\dot{N}_{r}|),
\end{aligned}
\end{equation}
where we have changed the parity of the condition on $|S|$ in the second sum, as by removing $t_{1}$ from the set $S$ we changed the parity of its cardinality.

Next, we want to apply  Corollary \ref{cor:bound_2} to bound the pairs of expressions in the above. This requires us to specify the contribution of $r$ explicitly as we need to find pairs of observables in which one acts with $A_{r;0}+A_{r;1}$ and the other with $A_{r;0}-A_{r;1}$. We start by excluding $r$ from the set $\dot{N}_{r}\setminus\{t_{1}\}$, similarly as we have done with $t_{1}$
\begin{equation}
\begin{aligned}
\langle I \rangle \leqslant&    \sum_{\substack{S\subseteq N_{r}\setminus\{t_{1}\}\\ |S|\textrm{ - even}}} \langle \tilde{f}_{e,S\cup\{r\}}\rangle+ \sum_{\substack{S\subseteq N_{r}\setminus \{t_{1}\}\\ |S|\textrm{ - odd}}} \langle \tilde{f}_{e,S}\rangle\\
&+ \sum_{\substack{S\subseteq N_{r}\setminus\{t_{1}\}\\ |S|\textrm{ - odd}}} \langle \tilde{f}_{o,S\cup\{r,t_{1}\}}\rangle+ \sum_{\substack{S\subseteq N_{r}\setminus \{t_{1}\}\\ |S|\textrm{ - even}}} \langle \tilde{f}_{o,S\cup\{t_{1}\}}\rangle+2^{|\dot{N}_{r}|-1}(n-|\dot{N}_{r}|)\\
=&\left(\sum_{\substack{S\subseteq N_{r}\setminus\{t_{1}\}\\ |S|\textrm{ - even}}} \langle \tilde{f}_{e,S\cup\{r\}}\rangle+\langle \tilde{f}_{o,S\cup\{t_{1}\}}\rangle\right)+\left(\sum_{\substack{S\subseteq N_{r}\setminus \{t_{1}\}\\ |S|\textrm{ - odd}}} \langle \tilde{f}_{o,S\cup\{r,t_{1}\}}\rangle+\langle \tilde{f}_{e,S}\rangle\right)+2^{|\dot{N}_{r}|-1}(n-|\dot{N}_{r}|).
\end{aligned}
\end{equation}
Next, our plan is to apply Corollary \ref{cor:bound_2} separately to each pair of terms in the first, as well as in the second sum. However, this cannot be applied directly to the expression above: we first need to split the expected values by maximising the above functional over all behaviours local across a given bipartition $Q|\overline{Q}$ \eqref{eq:product_behaviour}
\begin{equation}\label{eq: ineq_two_sums}
\begin{aligned}
\langle I \rangle \leqslant& \left(\sum_{\substack{S\subseteq N_{r}\setminus\{t_{1}\}\\ |S|\textrm{ - even}}} \langle (A_{r;0}+A_{r;1})\tilde{f}_{e,S\cup\{r\}}^{Q\setminus\{r\}}\rangle_{Q}\langle A_{l;\mu}\tilde{f}_{e,S\cup\{r\}}^{\overline{Q}\setminus\{l\}}\rangle_{\overline{Q}}+\langle (A_{r;0}-A_{r;1})\tilde{f}^{Q\setminus\{r\}}_{o,S\cup\{t_{1}\}}\rangle_{Q}\langle A_{l;1-\mu}\tilde{f}^{\overline{Q}\setminus\{l\}}_{o,S\cup\{t_{1}\}}\rangle_{\overline{Q}}\right)\\
&+\left(\sum_{\substack{S\subseteq N_{r}\setminus \{t_{1}\}\\ |S|\textrm{ - odd}}} \langle(A_{r;0}+A_{r;1}) \tilde{f}^{Q\setminus\{r\}}_{o,S\cup\{r,t_{1}\}}\rangle_{Q}\langle A_{l;1-\mu}\tilde{f}^{\overline{Q}\setminus\{l\}}_{o,S\cup\{r,t_{1}\}}\rangle_{\overline{Q}}+\langle (A_{r;0}-A_{r;1})\tilde{f}_{e,S}^{Q\setminus\{r\}}\rangle_{Q}\langle A_{l;\mu}\tilde{f}_{e,S}^{\overline{Q}\setminus\{l\}}\rangle_{\overline{Q}}\right)\\
&+2^{|\dot{N}_{r}|-1}(n-|\dot{N}_{r}|),
\end{aligned}
\end{equation}
where we have also explicitly written the observables of the parties $r$ and $l$, and $\mu=m\mod 2$. Let us begin by applying Corollary \ref{cor:bound_2} to the first term by taking $\alpha_{0}=\langle A_{l;\mu}\tilde{f}_{e,S\cup\{r\}}^{\overline{Q}\setminus\{l\}}\rangle_{\overline{Q}}$, $\alpha_{1}=\langle A_{l;1-\mu}\tilde{f}^{\overline{Q}\setminus\{l\}}_{o,S\cup\{t_{1}\}}\rangle_{\overline{Q}}$. To do so, we first need to prove the decomposition of the set of particles into six disjoint sets $D_{i,j}$. Let us construct $D_{0;0}$ and $D_{0;1}$ by taking the intersection of sets of parties acting with $A_{i;0}$ and $A_{i;1}$ respectively in Eq. \eqref{eq:f_e} and \eqref{eq:f_o}, and let us construct the rest of $D_{i;j}$ by taking appropriate differences of these sets. 
Written explicitly, they read
\begin{equation}
\begin{aligned}
D_{0;0}&=[S\cup\{r\}\cup T_{0}]\cap [S\cup \{t_{1}\}\cup T_{1}]\cap [Q\setminus\{r\}]=S\\
D_{1;0}&=([S\cup\{r\}\cup T_{0}]\setminus [S\cup \{t_{1}\}\cup T_{1}])\cap [Q\setminus\{r\}]=T_{0}\setminus\{l\}\\
D_{2;0}&=([S\cup \{t_{1}\}\cup T_{1}]\setminus [S\cup\{r\}\cup T_{0}])\cap [Q\setminus\{r\}]=\{t_{1}\}\cup T_{1}\setminus \{l\}\\
D_{0;1}&=\left(\underset{i\in S\cup\{r\}\cup T_{0}}{\Delta} N_{i}\right)\cap\left(\underset{i\in S\cup \{t_{1}\}\cup T_{1}}{\Delta} N_{i}\right)\cap (Q\setminus \{r\}),\\
D_{1;1}&=\left(\left(\underset{i\in S\cup\{r\}\cup T_{0}}{\Delta} N_{i}\right)\setminus\left(\underset{i\in S\cup \{t_{1}\}\cup T_{1}}{\Delta} N_{i}\right)\right)\cap (Q\setminus \{r\}),\\
D_{2;1}&=\left(\left(\underset{i\in S\cup \{t_{1}\}\cup T_{1}}{\Delta} N_{i}\right)\setminus\left(\underset{i\in S\cup\{r\}\cup T_{0}}{\Delta} N_{i}\right)\right)\cap (Q\setminus \{r\}),\\
\end{aligned}
\end{equation}
where to simplify the first three expressions, we used the following facts: $T_{0}$, $T_{1}$ are disjoint, the shortest path assumption, the assumptions $Q\cap \dot{N}_{r}=\dot{N}_{r}$, and $t_{i}\in Q$ for all $i\in\{0,\ldots, m-1\}$. From construction it easily follows that $D_{i;0}$ are disjoint for all $i\in\{0,1,2\}$ as well as $D_{i;1}$ for all $i\in\{0,1,2\}$. Since, by previous derivation, $\tilde{f}_{o,S\cup\{r\}}^{Q}$ and $\tilde{f}_{e,S\cup\{t_{1}\}}^{Q}$ are well-defined observables, 
and we have
\begin{equation}
\tilde{f}_{e,S\cup\{r\}}^{Q\setminus\{r\}}\approx \prod_{i\in D_{0;0}\cup D_{1;0}}A_{i;0}\prod_{j\in D_{0;1}\cup D_{1;1}}A_{j;1},\quad \tilde{f}_{o,S\cup\{t_{1}\}}^{Q\setminus\{r\}}\approx \prod_{i\in D_{0;0}\cup D_{2;0}}A_{i;0}\prod_{j\in D_{0;1}\cup D_{2;1}}A_{j;1},
\end{equation}
it follows that $D_{0;0},D_{0;1},D_{1;0},D_{1;1}$ are disjoint, and so are $D_{0;0},D_{0;1},D_{2;0},D_{2;1}$. Therefore, to show that all $D_{i;j}$ are disjoint, we need to show that the pairs $D_{1;0},D_{2;1}$ and $D_{1;1}D_{2;0}$ are disjoint. 

Let us begin from the pair $D_{1;0},D_{2;1}$. Our plan is to count the number of times each party in $D_{1;0}$ appears in the neighbourhoods that define $D_{2;1}$. If we can show that each party appears an even number of times in these neighbourhoods, then this implies that they do not appear in $D_{2;1}$ as a symmetric difference cancels out such contributions.

First, notice that the shortest path assumption implies that parties from $T_{0}\setminus \{l\}$ appear in exactly two neighbourhoods $N_{i}$ for $i\in \{t_{1}\}\cup T_{1}$, and parties from $T_{0}\setminus \{l,t_{2}\}$ do not appear in the neighbourhoods of $S$. And so $T_{0}\setminus \{l,t_{2}\}\notin D_{2;1}$. As for $t_{2}$ it appears in the neighbourhoods of $S$ $|S\cap N_{t_{2}}|$ many times, which can be either even or odd, depending on the graph $G$. However, it appears as many times in the neighbourhoods of $S\cup\{r\}\cup T_{0}$: $t_{2}$ is not connected to $r$ or any party from $T_{0}$ (by the shortest path assumption), and so the only contribution surviving is from the neighbourhoods of $S$, which as previously established, means that it appears $S\cap N_{t_{2}}$. Therefore, contributions from $T_{0}\setminus\{l,t_{2}\}$ are cancelled out by the symmetric difference, whereas $t_{2}$ gets cancelled out by the set difference. And so we have $D_{1;0}\cap D_{2;1}=\emptyset$.

Let us now turn to $D_{1;1}D_{2;0}$. Parties from $T_{1}\setminus \{l\}$ appear in the neighbourhoods of $T_{0}$ exactly two times by the shortest path assumption, and they do not appear in the neighbourhoods of $S\cup\{r\}$. As for $t_{1}$, it appears in the neighbourhoods of $S\cup\{t_{1}\}\cup T_{1}$ exactly $|S\cap N_{t_{1}}|$ times. On the other hand, $t_{1}$ appears in the neighbourhoods of $S\cup\{r\}\cup T_{0}$, $|S\cap N_{t_{1}}|+2$ times: $|S\cap N_{t_{1}}|$ times for $S$, once in $N_{r}$ and once in $N_{t_{2}}$, where as a reminder $t_{2}\in T_{0}$. Notice, that $|S\cap N_{t_{1}}|$ and $|S\cap N_{t_{1}}|+2$ have the same parity, and so the contribution from $t_{1}$ is cancelled out by the symmetric difference if $|S\cap N_{t_{1}}|$ is even, and by set difference if $|S\cap N_{t_{1}}|$ is odd. Therefore, $D_{2;0}\cap D_{1;1}=\emptyset$.

With this we have shown that all $D_{i;j}$ are disjoint, which allows us to use  Corollary \ref{cor:bound_2} to get 
\begin{equation}\label{eq: last_lemma_app_1}
\begin{aligned}
\langle (A_{r;0}+A_{r;1})\tilde{f}_{e,S\cup\{r\}}^{Q\setminus\{r\}}\rangle_{Q}\langle A_{l;\mu} \tilde{f}_{e,S\cup\{r\}}^{\overline{Q}\setminus\{l\}}\rangle_{\overline{Q}}+\langle (A_{r;0}-A_{r;1})\tilde{f}^{Q\setminus\{r\}}_{o,S\cup\{t_{1}\}}\rangle_{Q}\langle A_{l;1-\mu}\tilde{f}^{\overline{Q}\setminus\{l\}}_{o,S\cup\{t_{1}\}}\rangle_{\overline{Q}}\leqslant2.
\end{aligned}
\end{equation}

We now turn our attention to the second term in Eq. \eqref{eq: ineq_two_sums}. We once again want to apply Corollary \ref{cor:bound_2}, with $\alpha_{0}=\langle A_{l;1-\mu}\tilde{f}^{\overline{Q}\setminus\{l\}}_{o,S\cup\{r,t_{1}\}}\rangle_{\overline{Q}} $ and $\alpha_{1}=\langle A_{l;\mu}\tilde{f}^{\overline{Q}\setminus\{l\}}_{e,S}\rangle_{\overline{Q}}$. We once again construct $D_{0;0}$ and $D_{0;1}$ by taking the intersection of sets of parties acting with $A_{i;0}$ and $A_{i;1}$ respectively in Eq. \eqref{eq:f_e} and \eqref{eq:f_o}, and the rest of $D_{i;j}$ by taking appropriate differences of these sets
\begin{equation}
\begin{aligned}
D_{0;0}&=[S\cup\{r,t_{1}\}\cup T_{1}]\cap [S\cup  T_{0}]\cap [Q\setminus \{r\}]=S,\\
D_{1;0}&=([S\cup\{r,t_{1}\}\cup T_{1}]\setminus [S\cup T_{0}])\cap [Q\setminus \{r\}]=\{t_{1}\}\cup T_{1}\setminus\{l\},\\
D_{2;0}&=([S\cup T_{0}]\setminus [S\cup\{r,t_{1}\}\cup T_{1}])\cap [Q\setminus \{r\}]=T_{0}\setminus \{l\},\\
D_{0;1}&=\left(\underset{i\in S\cup\{r,t_{1}\}\cup T_{1}}{\Delta} N_{i}\right)\cap\left(\underset{i\in S\cup T_{0}}{\Delta} N_{i}\right)\cap (Q\setminus \{r\}),\\
D_{1;1}&=\left(\left(\underset{i\in S\cup\{r,t_{1}\}\cup T_{1}}{\Delta} N_{i}\right)\setminus\left(\underset{i\in S\cup T_{0}}{\Delta} N_{i}\right)\right)\cap (Q\setminus \{r\}),\\
D_{2;1}&=\left(\left(\underset{i\in S\cup T_{0}}{\Delta} N_{i}\right)\setminus\left(\underset{i\in S\cup\{r,t_{1}\}\cup T_{1}}{\Delta} N_{i}\right)\right)\cap (Q\setminus \{r\}).\\
\end{aligned}
\end{equation}
Then by the same logic as before, we need to show that the pairs $D_{1;0}, D_{2;1}$ and $D_{2;0},D_{1;1}$ are disjoint to show that all $D_{i;j}$ are disjoint.

First, let us analyse $D_{1;0}, D_{2;1}$. By the shortest path assumption, every party from $T_{1}\setminus \{l\}$ appears twice in the neighbourhoods of $T_{0}$ and does not appear in the neighbourhoods of $S$. As for $t_{1}$, it appears in every neighbourhood of $S$ and in the neighbourhood of $t_{2}\in T_{0}$. Since by Eq. \eqref{eq: ineq_two_sums} $|S|$ is odd, we get that $t_{2}$ appears an even number of times in the neighbourhoods of $S\cup T_{0}$. And so $D_{1;0}\cap D_{2;1}=\emptyset$.

Next, let us look at $D_{2;0},D_{1;1}$. Every party in $T_{0}\setminus\{l,t_{2}\}$ appears twice in the neighbourhoods of $T_{1}$ and never in the neighbourhoods of $S\cup\{r,t_{1}\}$. As for $t_{2}$, it appears in the neighbourhood of $t_{1}$, $t_{3}\in T_{1}$, and $|S\cap N_{t_{2}}|$ times in the neighbourhoods of $S$, giving the total of $|S\cap N_{t_{2}}|+2$ appearances. As this number can be even or odd depending on the graph $G$, we have to look at the neighbourhoods whose contributions we remove in $D_{1;1}$. There, $t_{2}$ appears $|S\cap N_{t_{2}}|$ times in the neighbourhoods of $S$, and never in the neighbourhoods of $T_{0}$. Therefore, the contribution from $t_{2}$ cancels out by symmetric difference if $|S\cap N_{t_{2}}|$ is even, and by set difference if $|S\cap N_{t_{2}}|$ is odd. Putting it all toghether we conclude that $D_{2;0}\cap D_{1;1}=\emptyset.$

With that, we are permitted to use Corollary \ref{cor:bound_2} to derive the following bound
\begin{equation}\label{eq: last_lemma_app_2}
\begin{aligned}
\langle(A_{r;0}+A_{r;1}) \tilde{f}^{Q\setminus\{r\}}_{o,S\cup\{r,t_{1}\}}\rangle_{Q}\langle A_{l;1-\mu}\tilde{f}^{\overline{Q}\setminus\{l\}}_{o,S\cup\{r,t_{1}\}}\rangle_{\overline{Q}}+\langle (A_{r;0}-A_{r;1})\tilde{f}_{e,S}^{Q\setminus\{r\}}\rangle_{Q}\langle A_{l;\mu}\tilde{f}_{e,S}^{\overline{Q}\setminus\{l\}}\rangle_{\overline{Q}}\leqslant2.
\end{aligned}
\end{equation}

Finally, applying Eq. \eqref{eq: last_lemma_app_1} and \eqref{eq: last_lemma_app_2} to Eq. \eqref{eq: ineq_two_sums} we get
\begin{equation}
\langle I\rangle \leqslant  2^{|N_{r}|-2}\cdot2 + 2^{|N_{r}|-2}\cdot2 + 2^{|\dot{N}_{r}|-1}(n-|\dot{N}_{r}|)=2^{|\dot{N}_{r}|-1}(n-|\dot{N}_{r}|+1),
\end{equation}
With this, we have shown that for all $Q|\overline{Q}$, the functional \eqref{eq: gmnl_ineq_clique_app} calculated over any behaviour local across $Q|\overline{Q}$ is bounded from above by $2^{|\dot{N}_{r}|-1}(n-|\dot{N}_{r}|+1)$, and so $\beta_{BL}\leqslant 2^{|\dot{N}_{r}|-1}(n-|\dot{N}_{r}|+1)$. To show that the equality holds, it is sufficient to consider a fully local strategy for which $\langle A_{i;j}\rangle =1$ for all $i\in [n]$ and $j\in\{0,1\}$.

With the bi-local bound established, we now compute the quantum bound of the functional \eqref{eq: gmnl_ineq_clique_app}. This can be easily calculated by the sum-of-squares decomposition
\begin{equation}
\begin{aligned}
2^{|\dot{N}_{r}|-1}(n-|\dot{N}_{r}|+\sqrt{2})\mathbb{1}-I=\sum_{\substack{S\subseteq \dot{N}_{r}\\|S|\textrm{ - odd}}}\frac{1}{\sqrt{2}}\left(\mathbb{1}-\frac{1}{\sqrt{2}}\tilde{f}_{S} \right)^{2}+2^{|\dot{N}_{r}|-2}\sum_{k\in [n]\setminus \dot{N}_{r}}\left(\mathbb{1}-\tilde{g}_{k}\right)^{2}.
\end{aligned}
\end{equation}
The only nontrivial part of showing that this equality holds is proving that all anticommutators $\{A_{r;0},A_{r;1}\}$ that originate from $(A_{r;0}\pm A_{r;1})^{2}=2\mathbb{1}\pm \{A_{r;0},A_{r;1}\}$ cancel out. However, this follows from the fact that there is an equal number of subsets $S$ such that $r\in S$ and that $r\notin S$, and so there is an equal number of anticommutators with a plus and minus sign.

Finally to show that $\beta_{Q}=2^{|\dot{N}_{r}|-1}(n-|\dot{N}_{r}|+\sqrt{2})$, it can be easily shown that this value is achievable by the graph state $\ket{G}$ with observables \eqref{eq: observable_choice}.
\end{proof}

\section{Utility of clique-based inequality}\label{app: lc}
Here, we discuss in detail the applicability of Theorem \ref{thm: clique}. While the condition required by this theorem seems very restrictive, it is very much not the case due to the local complementation equivalence between graphs. In Ref. \cite{PhysRevA.69.062311,Cabello_2009}, the authors classified all classes of connected graphs up to $8$ vertices that are equivalent under local complementation, i.e., they correspond to graph states inequivalent up to local unitaries. In this section, we analyse these classes and the applicability of Theorem \ref{thm: clique} to each of them. In doing so, we use the numbering system of graph classes from \cite{PhysRevA.69.062311,Cabello_2009}.

For the vast majority of class representatives, one cannot use Theorem \ref{thm: clique} to gain an advantage in $\beta_{Q}/\beta_{BL}$ ratio over Theorem \ref{thm: gmnl_standard}. To do so, we need a vertex $r$ whose neighbourhood $N_{r}$ is a complete induced subgraph and $|N_{r}|\geqslant 2$. However, it is also true that for the majority of the class representatives, one needs to apply only one local complementation to get a graph satisfying the conditions of Theorem \ref{thm: clique}. In Table \ref{tab: long_list}, we list all of such classes using the numbering system of graph classes from \cite{PhysRevA.69.062311,Cabello_2009}. Notably, we do not list the graph in a given class that achieves the highest ratio, but the one that achieves the highest ratio with a single application of local complementation.
\begin{table}[]
\begin{tabular}{|l|l|}
\hline
Vertex & Class number \\ \hline
1 & 2,3,5,8,9,17,18,20,40,43,46,58,70,87,89,90,100,128,131,132,138,141,145,146 \\ \hline
2 & 4,6,7,12,14,27,30,36,41,68,104,111,116,121,137,144 \\ \hline
3 & 13,28,29,31,33.34.35,37,42,92,95,101,102,103,106,108,112,115,117,118,119,120,130,133 \\ \hline
4 & 79,86,91,94,105,113,125 \\ \hline
5 & 39,93 \\ \hline
6 & 10,11,15,38,67,97,127,142 \\ \hline
7 & 21,22,23,24,25,26,44,96,129,140 \\ \hline
8 & 47,48,49,50,51,52,53,54,55,56,57,59,60,61,62,63,64,65,66,69,72,74,75,76,77,78,80,81,82,83,84,85,88,98,99,109,110,114,122,\\
& 126,136,139 \\ \hline
\end{tabular}
\caption{The list of graph classes, for which we can find a graph satisfying the conditions of Theorem \ref{thm: clique} with a single application of local complementation. The left column lists the vertex on which one needs to act with local complementation, which at the same time is also the choice for the rotated party, whereas the right column contains numbers of corresponding class representatives as per Ref. \cite{PhysRevA.69.062311,Cabello_2009}}\label{tab: long_list}
\end{table}

The only classes for which one cannot apply Theorem \ref{thm: clique} with a single application of local complementation are 1, 16, 19, 32, 45, 71, 73, 107, 123, 124, 134, 135, 143. Eight of them have to be treated with local complementation up to $3$ times to find a graph complying with Theorem \ref{thm: clique} (see Table \ref{tab: short_list}). However, for classes no. $1, 19,45,134$ and $135$ there is no sequence of local complementations that gives $|N_{r}|\geqslant2$. The class number $1$ is trivial, as it is the only connected graph of order $2$. As for the other four classes, they all originate from a graph of order $6$ (see Fig. \ref{fig: 8-graph} for the graphical representation). Interestingly, this graph of order $6$ (class no. $19$) is the only class of connected graphs up to four vertices for which the application of Theorem \ref{thm: clique} is strictly impossible, as for the other problematic cases there exists a vertex $r$ for which $|N_{r}|=1$, which makes the functional from Theorem \ref{thm: clique} equivalent to the one from Theorem \ref{thm: gmnl_standard}.

\begin{table}[]
\begin{tabular}{|l|l|l|}
\hline
Class number & Local complementation sequence & Rotated party \\ \hline
16 & 2,3,4 & 3 \\ \hline
32 & 5,6,7 & 1 \\ \hline
71 & 4,7,8 & 1 \\ \hline
73 & 4,7,8 & 1 \\ \hline
107 & 1,5 & 1 \\ \hline
123 & 2,3 & 7 \\ \hline
124 & 1,4,5 & 3 \\ \hline
143 & 2,7 &  1\\ \hline
\end{tabular}
\caption{The list of graph classes, for which to find a graph satisfying the conditions of Theorem \ref{thm: clique} we need to apply local complementation more than once. The first column lists the class number of the graph representative as per Ref. \cite{PhysRevA.69.062311,Cabello_2009}, the second column list the sequence of local complementation which gives the graph satisfying the conditions of Theorem \ref{thm: clique} (note that local complementation do not commute, so the order from left to right has to be preserved), and the last column lists the choice of a rotated party $r$ for this transformed graph.}\label{tab: short_list}
\end{table}

\begin{figure}[ht]
    \centering
    \includegraphics[width=0.25\linewidth]{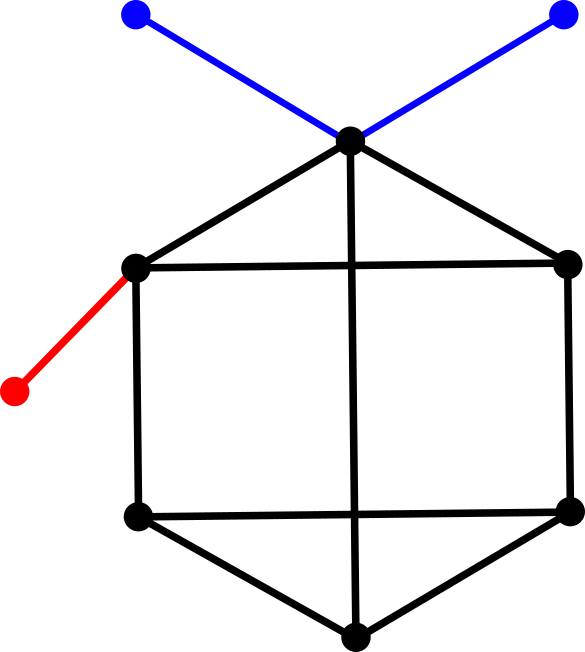}
    \caption{The representatives of classes of graphs for which finding a vertex $r$ with all its neighbours connected to each other and $|N_{r}|\geqslant 2$ is impossible. The class no. $19$ is the graph in black, the class no. $45$ is the graph in black plus one of the blue vertices, the class no. $134$ is the graph in black plus both of the blue vertices and the class no. $135$ is the graph in black plus one of the vertices in blue and the vertex in red.}
    \label{fig: 8-graph}
\end{figure}

\section{Numerical derivation of fidelity bounds}\label{app:fid}

Given an observed violation $\beta_n$ for an $n$-partite Bell inequality, a robust self-testing statement can be made by bounding the following quantity
\begin{align}\label{eq_app:defF}
    F_{\rm DI}(\beta_n)=\inf_{{\rho}\in \mathcal{S}(\beta_n)} \max_{\lbrace{\Lambda_i} \rbrace_{i=1}^n}F(\Lambda_{1}\otimes\ldots \Lambda_{n}({\rho}),\Phi^n),
\end{align}
with $\Lambda_i$, $i\in [n]$ is a quantum channel, $F(\rho,\sigma)=(\Tr|\sqrt{\rho}\sqrt{\sigma}|)^2$ the fidelity of quantum states and the infimum is taken over the set $\mathcal{S}(\beta_n)$ of all states that achieve violation $\beta_n$ for the given inequality.
For each $n$-partite Bell inequality of interest, we proceed by a numerical search that provides coefficients $f_n$ and $\mu_n$ such that the following inequality
\begin{align}\label{eq:Fstandard}
F_{\rm DI}(\beta_n)\geq f_n \beta_n-\mu_n,
\end{align}
applies.
In order to find such bounds, we employ the self-testing from operator inequality (STOPI) method introduced in~[33] 
and further developed in~[34]
. 
The method consists in choosing a particular extractable channel, i.e. fixing the maps $\DE{\Lambda_i}_{i=1}^n$ in Eq. \eqref{eq_app:defF}, and turning \eqref{eq:Fstandard} into an operator inequality that is valid for any quantum state. This is done by noting that for fixed quantum state and channels, the r.h.s. of Eq. \eqref{eq_app:defF} can be turned into
\begin{align}
\begin{split}
    F(\otimes_{i=1}^n \Lambda_i(\rho),\Phi_n)&= \Tr\de{\otimes_{i=1}^n\Lambda_i(\rho)\Phi_n}\\
    &=\Tr\de{\otimes_{i=1}^n\Lambda^\dag_i(\Phi_n)\rho}\\
    &=\Tr\de{K\rho}
    \end{split}
\end{align}
where $K:=\otimes_{i=1}^n \Lambda^\dag_i(\Phi_n)$.
Moreover $\beta_n=\Tr\de{W\rho}$, where $W$ is the Bell operator that represents the corresponding Bell inequality. Therefore a valid bound of the form \eqref{eq:Fstandard} is obtained for every $f_n$
, $\mu_n$ such that following operator
\begin{align}\label{eq:operator}
K-f_n W+\mu_n\id
\end{align}
is positive semidefinite.
For more details we refer the reader to \cite{Kaniewski2015analytic}.

We specialize our analysis to two different Bell inequalities of the form \eqref{eq: gmnl_ineq_general} for the star graph state (the $n$-vertex graph where vertex $1$ is connected to all the others). Notice that the star graph state is unitarily equivalent to the $n$-partite GHZ state. We derive numerical fidelity bounds for two families of inequalities for the graph, one where $r = 1$ and the other for $r = 2$. For completeness, we report their expression here as well: the two inequalities read

\begin{equation}\label{eq_app: NGHZ_gmnl_r=1}
\begin{aligned}
\langle (A_{1;0}+A_{1;1})A_{2;1} \ldots A_{n;1}\rangle 
&+ \sum_{i = 2}^n \langle (A_{1;0} - A_{1;1}) A_{i;0}\rangle \leq 2 (n-1), 
\end{aligned}
\end{equation}
for $r=1$ and
\begin{equation}\label{eq_app: NGHZ_gmnl_r=2}
\begin{aligned}
\langle A_{1;1}(A_{2;0}+A_{2;1})\rangle + \langle A_{1;0}&(A_{2;0}-A_{2;1})A_{3;1} \ldots A_{n;1}\rangle\\
&+ \sum_{i =2}^n 2 \langle (A_{1;1} A_{i;0}\rangle \leq 2(n-1),
\end{aligned}
\end{equation}
for $r=2$.
The maximal quantum violation of the two inequalities amounts to $\beta_Q^{r=1} = 2 \sqrt{(n-1)^2 + 1}$ and $\beta_Q^{r=2} = 2 \sqrt{2} + 2 (n-2)$ respectively.
In both cases, since the inequalities have binary inputs and outputs, the analysis can be restricted to qubits and projective measurements~ (see [33]
). Hence, the local observables $A_{i;x_i}$ can now be parametrized as
\begin{equation}\label{eq:measr}
A_{r;x_r} = \cos{\alpha_r}\, X + 
(-1)^{x_r} \sin{\alpha_r}\, Z, 
\end{equation}
and
\begin{equation}\label{eq:measi}
A_{i;x_i} = \cos{\alpha_i}\, H + 
(-1)^{x_i} \sin{\alpha_i}\, V  
\end{equation}
for $i \neq r$,
where $H = (X + Z)/\sqrt{2}$, $V = (X - Z)/\sqrt{2}$ and $\alpha_i \in [ 0,\pi/2] $. The different choice of parametrization is made so that the measurement leading to the maximal quantum violation
are included in the parametrization. In particular, for \eqref{eq_app: NGHZ_gmnl_r=2} they occur at $\alpha_i = \pi/4$ for all parties, while for \eqref{eq_app: NGHZ_gmnl_r=1} they occur at $\alpha_i = \pi/4$ for $i \neq r$ and at $\alpha_r = \arccos{(\frac{1}{\sqrt{(n-1)^2 +1}})}$ for $i = r$.

For inequality \eqref{eq_app: NGHZ_gmnl_r=2}, we fix the same extractable channels of Ref.~\cite{Kaniewski2015analytic}
, i.e., for each party $i$, $\Lambda_i^\dag$ is a dephasing channel determined by the angle $\alpha_i$ between the observables used for the inequality violation. More explicitly
\begin{equation}\label{eq_app:ex_channel}
\Lambda_i(x) [\rho] = \frac{1 + g(x)}{2} \rho + 
\frac{1 - g(x)}{2} \Gamma_i (x) \rho \Gamma_i (x)
\end{equation}
and the dependence on the measurement angle is encoded in the function $g(x) = (1+ \sqrt{2} )(\sin{x}+\cos{x} -1)$ together with the definition of $\Gamma(x) = M_i^{a}$ if $x \leq \pi /4$ and $\Gamma(x) = M_i^{b}$ if $x > \pi /4$. Lastly, we define $M_r^{a,b} = X,
Z$ and  $M_i^{a,b} = H,V$ for $i \neq r$.

For inequality \eqref{eq_app: NGHZ_gmnl_r=1}, we need a different choice of channels, that take into account that the measurement angle leading to the maximal quantum violation is different than $\pi/4$ for party $i = 1$. In particular, by following the measurement parametrisation of \eqref{eq:measr}, the maximal violation is obtained for $\alpha^{\text{opt}}(n) = \arcsin{(n/\sqrt{(n-1)^2 + 1})}$. Hence, we need to define an extraction channel that acts as the identity at the angle of maximal violation. We find the choice made in Ref. \cite{Coopmans2019robust} to be suitable for our inequalities as well. Define the effective angle as
\begin{equation}
h_{n}(x) := \begin{cases}
    \dfrac{\pi}{4} \dfrac{x}{\alpha^{\text{opt}}(n)} & \text{if } x \in [0, \alpha^{\text{opt}}(n)] \\[10pt]
    \dfrac{\pi}{2} - \dfrac{\pi}{4} \dfrac{\pi - 2x}{\pi - 2\alpha^{\text{opt}}(n)} & \text{if } x \in (\alpha^{\text{opt}}(n), \pi/2]
\end{cases}
\end{equation}
we then take the extraction channel for the first party to be $\Lambda_1(h_n(x))$, where the channel is defined as in \eqref{eq_app:ex_channel}.
Lastly, a lower bound on $F_{\rm DI}$ can then be obtained by (i) fixing $\mu_n$ as a function of $f_n$ in order to obtain fidelity one for maximal violation, then (ii) minimize \eqref{eq:operator} over the angles between the observables used by the parties for each value of $f_n$, and finally (iii) taking $f_n$ to be the threshold value for which the minimum eigenvalue of \eqref{eq:operator} becomes positive. 
\end{document}